%% file: main.tex
\documentclass[runningheads]{llncs}

\usepackage[margin=1in]{geometry}
\usepackage{amsmath,amssymb,mathtools}
\usepackage{booktabs}
\usepackage{graphicx}
\usepackage{float}
\usepackage{algorithm}
\usepackage{xspace}
\usepackage{needspace}
\usepackage[noend]{algpseudocode}
\usepackage{flafter}
\usepackage{tikz}
\usetikzlibrary{arrows.meta,fit,positioning,calc,decorations.pathreplacing,backgrounds}
\definecolor{figblue}{HTML}{4E79A7}
\definecolor{figorange}{HTML}{F28E2B}
\definecolor{figgreen}{HTML}{59A14F}
\definecolor{figgray}{HTML}{6B7280}
\definecolor{figlightblue}{HTML}{EAF2FB}
\definecolor{figlightgreen}{HTML}{EDF7ED}
\usepackage{hyperref}
\usepackage[bibliography=common,appendix=append]{apxproof}
\usepackage{thmtools}
\usepackage{thm-restate}

\newtheoremrep{lemma}[theorem]{Lemma}
\newtheorem{observation}[lemma]{Observation}
\title{On the Metric Distortion of the Minimum Spanning Tree Problem}

\author{
  Kian Ghassemi\thanks{
    Sharif University of Technology, Tehran, Iran.
  }
  \hspace{1.5em}
  Masoud Seddighin\thanks{
    Tehran Institute for Advanced Studies (TeIAS), Khatam University, Tehran, Iran.
  }
  \hspace{1.5em}
  Saeed Seddighin\thanks{
    Independent Researcher, Chicago, USA.
  }
}

\institute{}
\newcommand{\Vertices}{V}
\newcommand{\NumVertices}{n}
\newcommand{\Weight}{w}
\newcommand{\Profile}{\pi}
\newcommand{\ConsMetrics}[1]{\mathcal{D}(#1)}

\newcommand{\Metric}{d}
\newcommand{\MST}{\textsf{MST}\xspace}
\newcommand{\cost}{\textsf{cost}}
\newcommand{\distortion}{\textsf{dist}}

\newcommand{\OPT}{\textsf{OPT}}
\newcommand{\rank}{\textsf{rank}}

\newcommand{\GB}{\textsf{gb}}
\newcommand{\LB}{\textsf{lb}}

\newcommand{\firstV}{i}
\newcommand{\secondV}{j}
\newcommand{\thirdV}{k}
\newcommand{\fourthV}{t}

\newcommand{\Component}{C}

\newcommand{\RealTree}{\mathcal{T}}
\newcommand{\RealTreeEdges}{E_{\RealTree}}

\newcommand{\Dim}{D}
\newcommand{\Spanner}{S}
\newcommand{\SpannerDist}{d_{\Spanner}}

\newcommand{\Budget}{B}

\DeclareMathOperator*{\argmin}{arg\,min}
\begin{document}

\begingroup
\makeatletter
\renewcommand{\@fnsymbol}[1]{%
  \ensuremath{%
    \ifcase#1
    \or *
    \or \dagger
    \or \ddagger
    \or \mathsection
    \or \mathparagraph
    \or \|
    \fi
  }%
}
\makeatother
\maketitle
\endgroup
\pagestyle{plain}

\setcounter{footnote}{0}
\renewcommand{\thefootnote}{\arabic{footnote}}

\begin{abstract}
We study the metric distortion of the minimum spanning tree
(\MST) problem when numerical edge weights are unavailable and
each vertex reports only a ranking of the other vertices induced
by these weights. We first show, through a simple example, that
without assumptions on the weights, no ordinal algorithm guarantees
a finite approximation ratio. We therefore focus on general metrics
and present a Bor\r{u}vka-inspired algorithm using only ordinal
comparisons that achieves distortion $O(\log n)$. We complement
this result by showing that, for every $\varepsilon>0$, no
deterministic ordinal algorithm guarantees distortion at most
$3/2-\varepsilon$ on all finite metric instances.

We then investigate more informative settings in pursuit of
constant distortion, first considering two well-studied metric
classes. For tree metrics realized on the given vertex set,
a more careful variant of our initial algorithm achieves
distortion~$1$, recovering an optimal \MST{} from ordinal information
alone. For the broad class of doubling metrics, a different
algorithm based on \emph{endpoint cut stability} achieves distortion
$2^{O(D)}$, where $D$ is the doubling dimension; this is constant
for fixed $D$, independently of $n$. Finally, we study a query-based
extension allowing comparisons between two edge lengths without
revealing their numerical values. For general metrics with $n\geq2$,
fewer than $n$ edge-comparison queries suffice to achieve distortion
at most $2\log_2^* n-1$, which is near-constant even for very large inputs.
\end{abstract}

\input{Files/introduction}

\input{Files/Our-Result}

\input{Files/Further-Related-Work}

\input{Files/Preliminaries}
\input{Files/Lower-bound}
\input{Files/MST-general}
\input{Files/MST-tree}
\input{Files/MST-doubling}
\input{Files/Query}

\input{Files/Conclusion}

\bibliographystyle{plain}
\bibliography{refs}

\end{document}

%% file: Files/introduction.tex
\section{Introduction}

Distortion is a fundamental and well-established concept in social choice theory
\cite{procaccia2006distortion,anshelevich2021distortion,anshelevich2018approximating,anshelevich2017randomized,goel2017metric,gkatzelis2020resolving,kizilkaya2022plurality,kempe2020communication}. Roughly speaking, it measures the efficiency loss incurred when solving an optimization problem using only ordinal comparisons rather than the underlying cardinal values.

Although distortion was originally introduced to analyze voting rules~\cite{procaccia2006distortion}, the underlying paradigm of optimization from ordinal information has since found applications well beyond social choice. In particular, ordinal algorithms have been developed for matching, clustering, and other combinatorial optimization problems~\cite{anshelevich2015blind,anshelevich2016truthful,anari2023distortion};
 see~\cite{anshelevich2021distortion} for a broad overview of distortion in social choice. In this paper, we extend the study of distortion to the minimum spanning tree (\MST) problem.

Let us first describe the problem with an example. Consider the complete weighted graph shown in Figure~\ref{fig:intro-example}. Suppose the goal is to find a minimum spanning tree  of this graph. One can easily check that the total weight of the MST is 45; the edges of the \MST are highlighted in green. 

If all the weights are given, then finding the \MST is easy using classical algorithms like Kruskal's or Prim's~\cite{kruskal1956shortest,prim1957shortest}. However, suppose that instead of having the exact weights, we only have access to the rankings of the vertices. That is, for each vertex, we know the order of the other vertices by increasing incident edge weight, but not the actual numerical values. For our example, these rankings are:
\[
\begin{array}{rcl}
	1 &:& 4 \succ 2 \succ 3 \succ 5,\\
	2 &:& 1 \succ 4 \succ 3 \succ 5,\\
	3 &:& 1 \succ 2 \succ 4 \succ 5,\\
	4 &:& 1 \succ 2 \succ 3 \succ 5,\\
	5 &:& 1 \succ 2 \succ 4 \succ 3.
\end{array}
\]
The question is whether the \MST can be uniquely determined from such ordinal information alone.

\begin{figure}[t]
	\centering
	\begin{tikzpicture}[scale=0.8]
		% Vertices
		\node[circle, draw, fill=white, inner sep=2pt, font=\bfseries] (1) at (0, 2.5) {1};
		\node[circle, draw, fill=white, inner sep=2pt, font=\bfseries] (2) at (2.5, 0.5) {2};
		\node[circle, draw, fill=white, inner sep=2pt, font=\bfseries] (3) at (1.5, -2) {3};
		\node[circle, draw, fill=white, inner sep=2pt, font=\bfseries] (4) at (-1.5, -2) {4};
		\node[circle, draw, fill=white, inner sep=2pt, font=\bfseries] (5) at (-2.5, 0.5) {5};

		% Non-MST edges
		\draw[figgray] (2) -- node[below right, fill=white, inner sep=1.5pt] {20} (3);
		\draw[figgray] (2) -- node[below right, fill=white, inner sep=1.5pt] {18} (4);
		\draw[figgray] (2) -- node[above, fill=white, inner sep=1.5pt] {25} (5);
		\draw[figgray] (3) -- node[below, fill=white, inner sep=1.5pt] {22} (4);
		\draw[figgray] (3) -- node[below left, fill=white, inner sep=1.5pt] {30} (5);
		\draw[figgray] (4) -- node[below left, fill=white, inner sep=1.5pt] {28} (5);

		% MST edges — total weight = 45
		\draw[figgreen, line width=0.5pt] (1) -- node[above right, fill=white, inner sep=0.5pt] {10} (2);
		\draw[figgreen, line width=0.5pt] (1) -- node[below left, fill=white, inner sep=0.5pt] {12} (3);
		\draw[figgreen, line width=0.5pt] (1) -- node[below right, fill=white, inner sep=0.5pt] {8} (4);
		\draw[figgreen, line width=0.5pt] (1) -- node[above left, fill=white, inner sep=0.5pt] {15} (5);
	\end{tikzpicture}
	\caption{A complete weighted graph. The \MST edges are highlighted in green.}
	\label{fig:intro-example}
\end{figure}
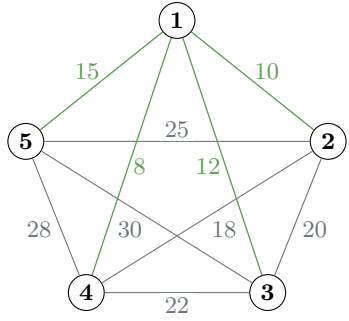

Unfortunately, the answer is \emph{no}. Figure~\ref{fig:intro-example-2} shows two weighted graphs that induce exactly the same ordinal information. At every vertex, the neighboring vertices are ranked identically:
\[
\begin{array}{rcl}
	1 &:& 2 \succ 4 \succ 3,\\
	2 &:& 1 \succ 3 \succ 4,\\
	3 &:& 4 \succ 2 \succ 1,\\
	4 &:& 3 \succ 1 \succ 2.
\end{array}
\]

Consequently, any algorithm that has access only to these rankings cannot distinguish between the two instances.
However,
 despite having identical ordinal information, the graphs have different minimum spanning trees. In Instance A, the unique \MST is
$
\{(1,2), (2,3), (3,4)\},
$
whereas in Instance B, the unique \MST is
$
\{(1,2), (1,4), (3,4)\}.$

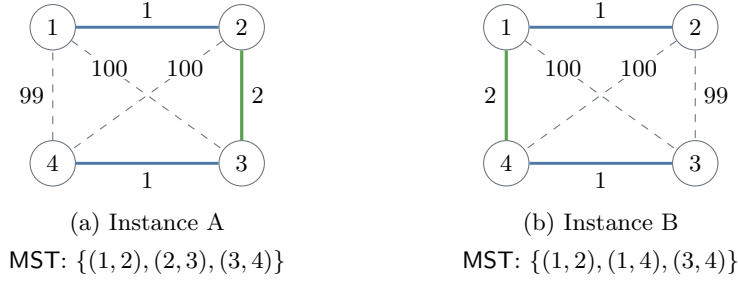
\begin{figure}[t]
\centering
\begin{tikzpicture}[
 vertex/.style={
   circle,draw=figgray,fill=white,
   minimum size=6mm,inner sep=1pt
 },
 fixed/.style={draw=figblue,line width=1.1pt},
 cheap/.style={draw=figgreen,line width=1.2pt},
 expensive/.style={draw=figgray,dashed},
 every node/.style={font=\small}
]
\begin{scope}
\node[vertex] (1) at (0,1.8) {$1$};
\node[vertex] (2) at (2.5,1.8) {$2$};
\node[vertex] (3) at (2.5,0) {$3$};
\node[vertex] (4) at (0,0) {$4$};

\draw[fixed] (1)--node[above] {$1$} (2);
\draw[fixed] (3)--node[below] {$1$} (4);
\draw[cheap] (2)--node[right] {$2$} (3);
\draw[expensive] (4)--node[left] {$99$} (1);
\draw[expensive] (1)--node[pos=.25,fill=white,inner sep=1pt] {$100$} (3);
\draw[expensive] (2)--node[pos=.25,fill=white,inner sep=1pt] {$100$} (4);

\node at (1.25,-.8) {(a) Instance A};
\node[font=\footnotesize] at (1.25,-1.3)
  {\MST: $\{(1,2),(2,3),(3,4)\}$};
\end{scope}

\begin{scope}[xshift=6cm]
\node[vertex] (1) at (0,1.8) {$1$};
\node[vertex] (2) at (2.5,1.8) {$2$};
\node[vertex] (3) at (2.5,0) {$3$};
\node[vertex] (4) at (0,0) {$4$};

\draw[fixed] (1)--node[above] {$1$} (2);
\draw[fixed] (3)--node[below] {$1$} (4);
\draw[expensive] (2)--node[right] {$99$} (3);
\draw[cheap] (4)--node[left] {$2$} (1);
\draw[expensive] (1)--node[pos=.25,fill=white,inner sep=1pt] {$100$} (3);
\draw[expensive] (2)--node[pos=.25,fill=white,inner sep=1pt] {$100$} (4);

\node at (1.25,-.8) {(b) Instance B};
\node[font=\footnotesize] at (1.25,-1.3)
  {\MST: $\{(1,2),(1,4),(3,4)\}$};
\end{scope}
\end{tikzpicture}
\caption{Two instances with identical vertex rankings but different
unique \MST{}s. In each instance, the blue and green edges form the \MST;
the remaining edges are dashed gray.}
\label{fig:intro-example-2}
\end{figure}

Note that the example in Figure~\ref{fig:intro-example-2} also highlights an additional difficulty. The edge weights can be modified; for instance, the values $99$ and $100$ can be replaced by arbitrarily large numbers without changing the induced vertex rankings. This means that the ordinal information remains the same even though the scale of the edge weights can change dramatically.
By increasing this gap, the difference between the weight of the optimal minimum spanning tree and that of the tree returned by any ordinal-based algorithm can be made arbitrarily large. As a consequence, even approximate optimality cannot be guaranteed when only ordinal information is available.
Therefore, additional structural assumptions on the edge weights are necessary in order to obtain meaningful approximation guarantees.

A natural and widely studied assumption that can make the problem tractable is that the vertices are embedded in a metric space, meaning that the edge weights satisfy the triangle inequality. This assumption has been extensively studied in the context of social choice theory, where it gives rise to the notion of \emph{metric distortion}~\cite{anshelevich2018approximating,anshelevich2017randomized,munagala2019improved,gkatzelis2020resolving,kizilkaya2022plurality,goel2017metric,kempe2020communication,charikar2022metric}.
Note that the example in Figure~\ref{fig:intro-example-2} does not satisfy the triangle inequality.

Motivated by this, the central question we study in this paper is:
\begin{quote}
	Can we compute a minimum spanning tree (or a good approximation of it) using only ordinal information when the underlying edge weights form a metric?
\end{quote}

%% file: Files/Our-Result.tex
\section{Our Results}
 We begin by proving a lower bound, showing that every deterministic ordinal algorithm has distortion at least  $\frac{3}{2}-o(1)$ as the number of vertices grows.
 
\begin{restatable}{restheorem}{mstlower}
\label{thm:mst-lower}
For every $\varepsilon>0$, no deterministic ordinal algorithm
for the minimum spanning tree problem guarantees distortion
at most $\frac{3}{2}-\varepsilon$ on all finite metric instances.
\end{restatable}

Recall that a lower bound of $3$ is known for deterministic
single-winner voting under metric preferences~\cite{anshelevich2018approximating}.
Our lower-bound proof follows the standard indistinguishability framework. We construct several metric instances that induce identical ordinal preferences but have different optimal minimum spanning trees.
Starting with an example for $n=4$, we recursively combine two smaller
instances to obtain an instance on $n=2^h$ vertices, using consecutive
Fibonacci numbers to define the distances at each level.
Any deterministic ordinal algorithm must output the same tree on all
instances with the same preferences, and at least one instance forces
its distortion to approach $\frac{3}{2}$ as $n$ grows.
This establishes Theorem~\ref{thm:mst-lower}.

We next present Algorithm~\ref{alg:mst-general} and show that it achieves an $O(\log n)$ distortion for the minimum spanning tree problem in metric spaces.

\begin{restatable}{restheorem}{mstgeneral}
\label{thm:mst-general}
Algorithm~\ref{alg:mst-general} is an ordinal algorithm for the minimum spanning tree problem with distortion $O(\log n)$ on general metrics.
\end{restatable}

Algorithm~\ref{alg:mst-general} is inspired by the contraction framework of \textsc{Borůvka}~\cite{boruuvka1926jistem} (see \cite{nevsetvril2001otakar} for an English translation). Initially, each vertex forms its own component. In every round, each component proposes an edge connecting it to another component. An acyclic subset of the proposed edges is then selected, added to the solution, and the resulting connected components are contracted. The process repeats until only a single component remains.

In Borůvka's algorithm, each component proposes its cheapest outgoing edge. Since our algorithm has access only to ordinal preferences rather than the actual edge weights, identifying the cheapest outgoing edge is impossible. Instead, we design a rule that carefully selects an outgoing edge using only the ordinal rankings.
The key ingredient of our analysis is a charging argument. We show that the total weight of the edges added in any single contraction round is at most a constant factor of the weight of an optimal minimum spanning tree. Moreover, every round reduces the number of components by a constant factor. Consequently, the algorithm terminates after $O(\log n)$ rounds, yielding an overall metric distortion of $O(\log n)$.

The significant gap between our lower and upper bounds naturally raises the question of whether better distortion guarantees can be achieved under more restricted metric assumptions. We investigate this question for two important classes of structured metrics: tree metrics and metrics of bounded doubling dimension.

We first consider tree metrics, which provide a canonical and highly structured setting in metric algorithms and have been extensively used as approximations of general metrics~\cite{bartal1996probabilistic,bartal1998approximating,fakcharoenphol2003tight}. In this case, the vertices are embedded in a weighted tree, and the distance between any two vertices is the length of the unique path connecting them in the tree. We then turn to metrics of bounded doubling dimension, which impose
a different structural restriction from tree metrics while still
allowing guarantees independent of the number of vertices.

For tree metrics, we obtain a tight upper bound. We show that a simple modification of our contraction-based algorithm constructs a spanning tree whose cost matches that of the optimal minimum spanning tree. Consequently, the algorithm achieves distortion \(1\), which is optimal.

\begin{restatable}{restheorem}{msttree}
\label{thm:mst-tree}
There is an ordinal algorithm for the minimum spanning tree problem with
distortion $1$ on tree metrics realized on the given vertex set.
\end{restatable}

Our algorithm still follows the same contraction framework, but it slightly changes how proposed edges are chosen. When a component proposes an edge to a vertex outside the component, we redirect the proposal to a better endpoint inside the component, using only the ordinal rankings. 
This redirection is the key idea that makes the analysis work. Although the algorithm does not know the hidden tree or the actual distances, we can show that every positive-cost edge added by the algorithm corresponds to a different positive edge of the underlying tree metric. The existence of zero-cost edges makes the claim nontrivial,
requiring us to analyze different cases. The algorithm
need not reproduce the exact edge set of an optimal \MST,
as its output may contain different zero-cost edges, but these
do not affect the total cost. Therefore, the returned spanning
tree has exactly the cost of an optimal \MST, giving distortion $1$.

We next consider metrics of bounded doubling dimension. A metric
has doubling dimension at most $\Dim$ if every ball of radius $r$
can be covered by at most $2^{\Dim}$ balls of radius $r/2$.
This class includes fixed-dimensional Euclidean metrics and is
incomparable with tree metrics: tree metrics can have arbitrarily
large doubling dimension, while bounded-doubling metrics need not
admit a tree representation. We show that bounded doubling dimension
allows distortion guarantees independent of $n$.

\begin{restatable}{restheorem}{mstdoubling}
\label{thm:mst-doubling}
There is an ordinal algorithm for the minimum spanning tree problem
with distortion $2^{O(\Dim)}$ on metrics of doubling dimension $\Dim$.
\end{restatable}

Our approach uses a spanning-tree property called
\emph{endpoint cut stability}, which requires that, when removing
any tree edge splits the vertices into two sets, each endpoint
prefers the other endpoint to every other vertex in the opposite set. Algorithm~\ref{alg:mst-doubling} constructs such a tree
from the ordinal profile alone using a Kruskal-style procedure.
We compare its cost with that of a light $(1+\epsilon)$-spanner:
a graph whose shortest-path distances approximate the metric
distances within a factor of $1+\epsilon$, and whose total edge
weight is bounded relative to the \MST cost.
Using the spanner construction of Borradaile, Le, and
Wulff-Nilsen~\cite{BorradaileLeWulffNilsen2019}, we obtain distortion
$2^{O(\Dim)}$, which is constant for every fixed $\Dim$.

Finally, we return to general metrics and allow a limited amount
of information beyond the ordinal profile. Motivated by the study
of information-distortion tradeoffs through cardinal value queries
by Amanatidis et al.~\cite{amanatidis2022few}, we ask how much
additional information suffices to improve distortion.
Interestingly, our algorithm never needs to learn a numerical
distance; it uses only \emph{edge-comparison queries}, each asking
whether $d(\firstV,\secondV)\le d(\thirdV,\fourthV)$.
Comparisons between edges sharing an endpoint are already available
from the ordinal profile, so additional queries are needed only
for edges with no common endpoint.

Our algorithm proceeds through a sequence of Bor\r{u}vka phases
on progressively smaller sets of representatives. Within each
phase, ordinal rankings identify the nearest outgoing neighbor
of each representative, and edge comparisons determine the
cheapest outgoing edge of each current component. We perform
several rounds on the same representative set before selecting
new representatives for the merged components. As the number
of representatives decreases, we can afford increasingly many
rounds per phase while keeping the total number of queries linear.
Each phase adds only a constant multiple of the optimal \MST cost,
and the phase schedule ensures termination after
$O(\log^{\ast} n)$ phases.

Thus, only $O(n)$ additional comparisons suffice to achieve
$O(\log^{\ast} n)$ distortion, without access to cardinal values.
Here $\log^{\ast} n$ is the number of times the logarithm must
be applied before the result is at most one; it grows so slowly
that it remains a small constant even for extraordinarily large
values of $n$.

\begin{restatable}{restheorem}{mstquery}
\label{thm:mst-query}
There is a deterministic algorithm for the minimum spanning tree
problem that, given the ordinal profile and using a total of
$O(n)$ additional edge-comparison queries, achieves distortion
$O(\log^{\ast} n)$ on general metrics.
\end{restatable}

%% file: Files/Further-Related-Work.tex
\section{Further Related Work}

The notion of \emph{distortion} was introduced by Procaccia and Rosenschein~\cite{procaccia2006distortion} and has since become a central concept in computational social choice.  Most of this literature focuses on selecting one or more alternatives based solely on these rankings~\cite{caragiannis2011voting,boutilier2012optimal,anshelevich2017randomized,anshelevich2018approximating,goel2017metric,munagala2019improved,gkatzelis2020resolving,anshelevich2021distortion,caragiannis2022metric}.

A substantial body of work has developed the metric framework for voting rules, analyzing how well they approximate optimal social cost under ordinal information~\cite{caragiannis2011voting,boutilier2012optimal,caragiannis2022metric}. We refer to the survey by Anshelevich, Filos-Ratsikas, Shah, and Voudouris~\cite{anshelevich2021distortion} for a comprehensive overview. This line of research has established tight distortion bounds, lower bounds, and fairness guarantees for both deterministic and randomized mechanisms~\cite{anshelevich2018approximating,anshelevich2017randomized,goel2017metric,munagala2019improved,gkatzelis2020resolving,kizilkaya2022plurality}. More recent work has further expanded the framework by considering settings with restricted or partial preference information, limited preference elicitation, distributed information, and average-case analyses~\cite{anagnostides2022metric,anshelevich2022distortion,amanatidis2021peeking,amanatidis2024don,amanatidis2022few,caragiannis2024beyond,ebadian2025every}.

Beyond voting, the distortion framework has also been extended to other optimization problems, including clustering and matching. In this broader setting—often referred to as ordinal approximation—algorithms again operate only on preference rankings induced by an underlying metric. Anshelevich and Sekar~\cite{anshelevich2015blind} initiated this direction by studying matching and clustering with purely ordinal information, and further explored truthful mechanisms in related ordinal settings~\cite{anshelevich2016truthful}. Abramowitz and Anshelevich~\cite{abramowitz2018utilitarians} studied ordinal approximation for graph problems with maximization objectives. Their framework
applies to several maximum-utility graph problems and implies, in particular, a
$2$-approximation for Maximum Spanning Tree using only ordinal information. Subsequent work has considered query-enhanced models, where algorithms may access a limited number of cardinal value queries beyond rankings, for matching and related combinatorial problems~\cite{amanatidis2022few,amanatidis2024don}. These results are closely related in spirit to ours, as they investigate what can be achieved from ordinal information, possibly augmented with minimal extra queries.

Recently, Anari, Charikar, and Ramakrishnan~\cite{anari2023distortion} study bipartite metric matching with $n$ agents and $n$ items, where each agent ranks items according to their metric distances. They prove an $O(n^2)$ upper bound for deterministic mechanisms, as well as an $\Omega(\log n)$ lower bound for any ordinal mechanism. In a different direction, Filos-Ratsikas, Gkatzelis, Latifian, Rewinski, and Voudouris~\cite{filos2025optimal} exploit the structure of the real line to obtain stronger guarantees for matching in one-dimensional metrics. Filos-Ratsikas and Kalantzis~\cite{filos2026distortion} further study distortion in stable matching settings, where the outcome is additionally required to satisfy stability constraints.

There is also a connection to classical approximation algorithms.
The $O(\log n)$ distortion bound we obtain for general metrics
also follows from the nearest-neighbor heuristic for metric
TSP~\cite{rosenkrantz1977analysis}.
This heuristic uses only ordinal information; removing its
closing edge yields a spanning tree, and the optimal tour
costs at most twice the optimal \MST.
Its logarithmic worst-case bound persists even in the
Euclidean plane~\cite{hougardy2015nearest}, and it need not
be optimal even on star metrics.\footnote{On a star with
$n-1$ unit-length edges, starting at the center gives an
NNH spanning path of cost $2n-3$, while the optimal \MST{}
costs $n-1$.}
Our Bor\r{u}vka-style framework supports extensions achieving
exact optimality on tree metrics realized on the given vertex
set and near-constant distortion with linear query complexity.
Whether it can also improve the logarithmic bound for general
metrics remains open.

%% file: Files/Preliminaries.tex
% ============================================================
\section{Preliminaries}
% ============================================================
Let $G=(\Vertices,\Metric)$ be a complete weighted graph, where $\Vertices$ is a finite set of $\NumVertices$ vertices and $\Metric$ is a weight function that assigns a nonnegative weight to every edge of the complete graph. In this paper, the edge weights are assumed to be induced by a metric, that is, the weights satisfy the following properties:
\begin{itemize}
    \item \emph{Reflexivity:} $\Metric(\firstV,\firstV)=0$.
    \item \emph{Symmetry:} $\Metric(\firstV,\secondV)=\Metric(\secondV,\firstV)$.
    \item \emph{Triangle inequality:} $\Metric(\firstV,\thirdV)\le \Metric(\firstV,\secondV)+\Metric(\secondV,\thirdV)$.
\end{itemize}
 We allow $\Metric(\firstV,\secondV)=0$ for distinct elements $\firstV\neq \secondV$; thus, throughout the paper, we use the term metric in this relaxed sense.

%%% part 2 
Our goal is to compute a minimum spanning tree of 
$G=(\Vertices,\Metric)$. However, the algorithm does not have access to the
metric $\Metric$. Instead, it is given only the \emph{ordinal profile}
$
\Profile(\Metric)=\bigl(\pi_{\firstV}(\Metric)\bigr)_{\firstV\in\Vertices},
$
that is consistent with $\Metric$. Here, for each vertex $\firstV\in\Vertices$, $\pi_{\firstV}(\Metric)$ is a strict ranking of the vertices in
$\Vertices\setminus\{\firstV\}$ that is consistent with their distances from
$\firstV$.
More precisely, consistency with $\Metric$ means that, for any two distinct vertices
$\secondV,\thirdV\in\Vertices\setminus\{\firstV\}$,
$\secondV \succ_{\firstV} \thirdV
$ implies $
\Metric(\firstV,\secondV)
\leq \Metric(\firstV,\thirdV),$
where $\secondV \succ_{\firstV} \thirdV$ means that $\secondV$ appears before
$\thirdV$ in the ranking $\pi_{\firstV}(\Metric)$. 
We denote by
$\rank_{\firstV}(\secondV)$\footnote{Formally, $\rank_{\firstV}(\secondV)$ is induced by the metric $\Metric$ through the ranking $\pi_{\firstV}(\Metric)$. When the metric is clear from context, we write $\rank_{\firstV}(\secondV)$  instead of $\rank^{\Metric}_{\firstV}(\secondV)$.}  the position of $\secondV$ in
$\pi_{\firstV}(\Metric)$ (smaller rank indicates a higher preference).

For a subset $S$ of edges, the cost of $S$ is defined as 
\[
\cost_{\Metric}(S)
=
\sum_{(\firstV,\secondV)\in S} \Metric(\firstV,\secondV).
\]

Given an algorithm $A$ for our problem and an ordinal profile $\Profile$, we
denote by $A(\Profile)$ the spanning tree returned after running $A$ on profile
$\Profile$. The distortion of $A$ on profile $\Profile$ is defined as
$$ \distortion(A,\Profile) = \sup_{\substack{\Metric\in\ConsMetrics{\Profile}}}
\frac{
\cost_{\Metric}\bigl(A(\Profile)\bigr)
}{
\cost_{\Metric}\bigl(\MST_{G}\bigr)
},
$$
where $\ConsMetrics{\Profile}$ denotes the set of all metrics consistent with $\Profile$,
and $\MST_{G}$ denotes a minimum spanning tree of $G = (V,d)$. The distortion of algorithm $A$ is defined as its worst-case distortion over all
ordinal profiles, that is,
$$\distortion(A) =\sup_{\Profile} \distortion(A,\Profile).$$

Finally, in the following we define two notions that are frequently used in our analysis.
\begin{definition}
Let $\Component\subsetneq \Vertices$ be a component, and let $\firstV\in\Component$. The globally best outside vertex of $\firstV$ with respect to $\Component$ is
\[
\GB_{\Component}(\firstV)
=
\argmin_{\secondV\in \Vertices\setminus \Component} \rank_{\firstV}(\secondV).
\]
That is, $\GB_{\Component}(\firstV)$ is the highest-ranked vertex of $\firstV$ outside its current component.
\end{definition}

\begin{definition}
Let $\Component\subseteq \Vertices$ be a component, and let $\firstV\notin\Component$. The locally best vertex of $\firstV$ in $\Component$ is
\[
\LB_{\Component}(\firstV)
=
\argmin_{\secondV\in\Component} \rank_{\firstV}(\secondV).
\]
That is, $\LB_{\Component}(\firstV)$ is the highest-ranked vertex of $\firstV$ among the vertices of $\Component$.
\end{definition}

%% file: Files/Lower-bound.tex
\section{Lower Bound}
\label{sec:mst-lower-bound}

We show that, for every $\varepsilon>0$, no deterministic ordinal
algorithm guarantees distortion at most $3/2-\varepsilon$ on all
finite metric instances. We begin with two metrics on four vertices
that give the same ordinal rankings but have different optimal
spanning trees. Repeating this idea recursively on $n=2^h$ vertices,
with distances drawn from the Fibonacci sequence, gives a lower
bound approaching $3/2$.

Consider the metrics $d_1$ and $d_2$ in
Figure~\ref{fig:mst-lower-bound-core}. In both, the edges
$(a,b)$ and $(c,d)$ have length $1$. In $d_1$, the edge $(a,c)$
also has length $1$, while the remaining edges have length $2$.
In $d_2$, the edge $(b,d)$ has length $1$ instead.
Both metrics satisfy the triangle inequality and are compatible
with the rankings
\[
a: b\succ c\succ d,\qquad
b: a\succ d\succ c,\qquad
c: d\succ a\succ b,\qquad
d: c\succ b\succ a,
\]
where equal distances are ordered as displayed.
Each metric has a unique MST of cost $3$, but these trees differ.
Every other spanning tree has cost at least $4$.
A deterministic ordinal algorithm must return the same tree
on both metrics, so its distortion is at least $4/3$ on one of them.

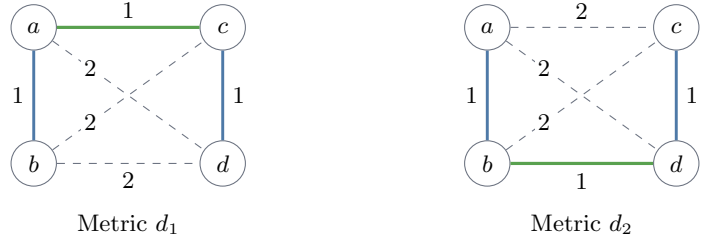
\begin{figure}[t]
\centering
\begin{tikzpicture}[
 vertex/.style={
   circle,draw=figgray,fill=white,
   minimum size=6mm,inner sep=1pt
 },
 fixed/.style={draw=figblue,line width=1.1pt},
 cheap/.style={draw=figgreen,line width=1.2pt},
 expensive/.style={draw=figgray,dashed},
 every node/.style={font=\small}
]
\begin{scope}
\node[vertex] (a) at (0,1.8) {$a$};
\node[vertex] (b) at (0,0) {$b$};
\node[vertex] (c) at (2.5,1.8) {$c$};
\node[vertex] (d) at (2.5,0) {$d$};
\draw[fixed] (a)--node[left] {$1$} (b);
\draw[fixed] (c)--node[right] {$1$} (d);
\draw[cheap] (a)--node[above] {$1$} (c);
\draw[expensive] (b)--node[below] {$2$} (d);
\draw[expensive] (a)--node[pos=.25,fill=white,inner sep=1pt] {$2$} (d);
\draw[expensive] (b)--node[pos=.25,fill=white,inner sep=1pt] {$2$} (c);
\node at (1.25,-.8) {Metric $d_1$};
\end{scope}
\begin{scope}[xshift=6cm]
\node[vertex] (a) at (0,1.8) {$a$};
\node[vertex] (b) at (0,0) {$b$};
\node[vertex] (c) at (2.5,1.8) {$c$};
\node[vertex] (d) at (2.5,0) {$d$};
\draw[fixed] (a)--node[left] {$1$} (b);
\draw[fixed] (c)--node[right] {$1$} (d);
\draw[expensive] (a)--node[above] {$2$} (c);
\draw[cheap] (b)--node[below] {$1$} (d);
\draw[expensive] (a)--node[pos=.25,fill=white,inner sep=1pt] {$2$} (d);
\draw[expensive] (b)--node[pos=.25,fill=white,inner sep=1pt] {$2$} (c);
\node at (1.25,-.8) {Metric $d_2$};
\end{scope}
\end{tikzpicture}
\caption{Two metrics compatible with the same ordinal rankings.
In each panel, the blue and green edges form the unique MST,
of cost $3$. Dashed gray edges have length $2$.}
\label{fig:mst-lower-bound-core}
\end{figure}

The same uncertainty can be repeated within larger sets of vertices.
Fix $h\ge2$, let $n=2^h$, and label the vertices $0,\ldots,n-1$.
Group consecutive vertices into pairs, combine consecutive pairs
into sets of four, then combine these into sets of eight,
and continue until all vertices belong to one set.
We call each set formed in this process a \emph{block};
its size is its number of vertices.
Thus every block of size at least four consists of two
consecutive blocks of equal size.
Figure~\ref{fig:lb-hierarchy} illustrates this arrangement for $n=8$.

Let $F_1=F_2=1$ and $F_{\ell+1}=F_\ell+F_{\ell-1}$ for $\ell\ge2$.
Assign distance $1$ within each pair.
Now consider a block $C$ of size $2^\ell$, where $\ell\ge2$.
Write its two halves as $A,B$, and split these into
$A_0,A_1$ and $B_0,B_1$, respectively, in increasing order
of vertex labels. For a block of size four,
these four smaller sets are singletons.

At $C$, choose $s_C\in\{0,1\}$.
Edges between $A_{s_C}$ and $B_{s_C}$ receive length $F_\ell$,
while all other edges between $A$ and $B$ receive length $F_{\ell+1}$.
Thus, for $u\in A$ and $v\in B$,
\[
d_s(u,v)=
\begin{cases}
F_\ell, & u\in A_{s_C}\text{ and }v\in B_{s_C},\\
F_{\ell+1}, & \text{otherwise}.
\end{cases}
\]
Here $s$ records the choices made at all blocks of size
at least four. Distances within each half are determined
by the same rule at smaller blocks.
Together with symmetry and $d_s(u,u)=0$, this specifies every
distance: the distance between two distinct vertices is assigned
at the smallest block containing both.

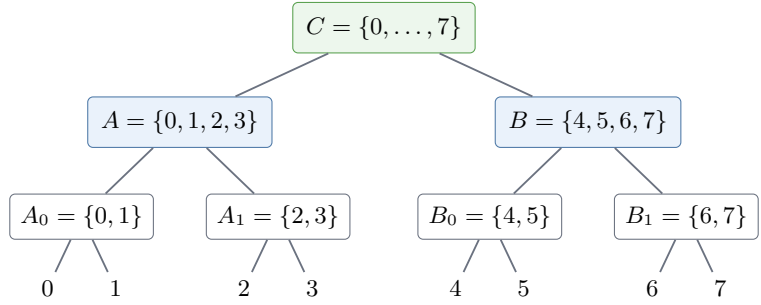
\begin{figure}[t]
\centering
\begin{tikzpicture}[
 x=1cm,y=1cm,
 block/.style={
   draw=figblue,fill=figlightblue,
   rounded corners=2pt,inner sep=5pt
 },
 pair/.style={
   draw=figgray,fill=white,
   rounded corners=2pt,inner sep=4pt
 },
 link/.style={draw=figgray,line width=.7pt},
 every node/.style={font=\small}
]
\node[block,draw=figgreen,fill=figlightgreen]
 (root) at (0,0) {$C=\{0,\ldots,7\}$};
\node[block] (a) at (-2.7,-1.25) {$A=\{0,1,2,3\}$};
\node[block] (b) at (2.7,-1.25) {$B=\{4,5,6,7\}$};
\node[pair] (a0) at (-4,-2.5) {$A_0=\{0,1\}$};
\node[pair] (a1) at (-1.4,-2.5) {$A_1=\{2,3\}$};
\node[pair] (b0) at (1.4,-2.5) {$B_0=\{4,5\}$};
\node[pair] (b1) at (4,-2.5) {$B_1=\{6,7\}$};
\draw[link]
 (root)--(a) (root)--(b)
 (a)--(a0) (a)--(a1)
 (b)--(b0) (b)--(b1);
\foreach \name/\xx/\ll/\rr in {
 a0/-4/0/1,a1/-1.4/2/3,b0/1.4/4/5,b1/4/6/7
}{
 \node (l\name) at (\xx-.45,-3.45) {$\ll$};
 \node (r\name) at (\xx+.45,-3.45) {$\rr$};
 \draw[link] (\name)--(l\name) (\name)--(r\name);
}
\end{tikzpicture}
\caption{The hierarchy for $n=8$.
At the largest block, either the edges between $A_0$ and $B_0$
or those between $A_1$ and $B_1$ receive length $2$.
The remaining edges between $A$ and $B$ receive length $3$.
The choices inside $A$ and $B$ determine their internal distances.}
\label{fig:lb-hierarchy}
\end{figure}

For $n=8$, consider choosing $s_C=1$ at all three blocks
of size at least four.
The pair edges $(0,1),(2,3),(4,5),$ $(6,7)$ have length $1$.
Within $\{0,1,2,3\}$, the edge $(1,3)$ has length $1$,
and the other edges between $\{0,1\}$ and $\{2,3\}$
have length $2$. Similarly, $(5,7)$ has length $1$,
and the other edges between $\{4,5\}$ and $\{6,7\}$
have length $2$.
Between the two blocks of size four, the edges
$(2,6),(2,7),(3,6),(3,7)$ have length $2$,
and all remaining edges have length $3$.

The two trees in Figure~\ref{fig:lb-eight-trees} illustrate
how these choices affect the cost.
Adding $(1,3),(5,7),(2,6)$ to the four pair edges gives cost
$4+1+1+2=8$.
Adding $(0,2),(4,6),(0,4)$ instead gives cost
$4+2+2+3=11$.
We will show that the first tree is optimal and that,
for every fixed output tree, some choice of $s$ forces
cost at least $11$.

\begin{figure}[t]
\centering
\begin{tikzpicture}[
 vertex/.style={
   circle,draw=figgray,fill=white,
   minimum size=5mm,inner sep=1pt,font=\scriptsize
 },
 fixed/.style={draw=figblue,line width=1pt},
 good/.style={draw=figgreen,line width=1.1pt},
 bad/.style={draw=figorange,line width=1.1pt},
 every node/.style={font=\small}
]
\begin{scope}
\foreach \v/\x/\y in {
 0/0/1.8,1/0/0,2/1.2/1.8,3/1.2/0,
 4/2.8/1.8,5/2.8/0,6/4/1.8,7/4/0
}
 \node[vertex] (v\v) at (\x,\y) {$\v$};
\foreach \a/\b in {0/1,2/3,4/5,6/7}
 \draw[fixed] (v\a)--(v\b);
\draw[good] (v1)--node[below] {$1$} (v3);
\draw[good] (v5)--node[below] {$1$} (v7);
\draw[good] (v2) to[bend left=32] node[above] {$2$} (v6);
\node at (2,-.85) {MST: cost $8$};
\end{scope}
\begin{scope}[xshift=6cm]
\foreach \v/\x/\y in {
 0/0/1.8,1/0/0,2/1.2/1.8,3/1.2/0,
 4/2.8/1.8,5/2.8/0,6/4/1.8,7/4/0
}
 \node[vertex] (v\v) at (\x,\y) {$\v$};
\foreach \a/\b in {0/1,2/3,4/5,6/7}
 \draw[fixed] (v\a)--(v\b);
\draw[bad] (v0)--node[below] {$2$} (v2);
\draw[bad] (v4)--node[below] {$2$} (v6);
\draw[bad] (v0) to[bend left=32] node[above] {$3$} (v4);
\node at (2,-.85) {A spanning tree of cost $11$};
\end{scope}
\end{tikzpicture}
\caption{Two spanning trees in the same metric on eight vertices.
Each blue edge has length $1$.
The green and orange edges show the additional connections
and their lengths.}
\label{fig:lb-eight-trees}
\end{figure}
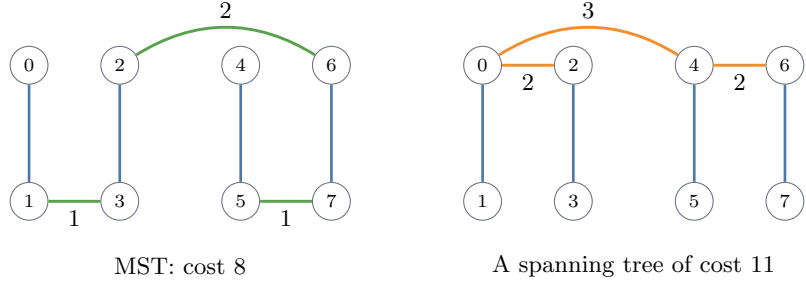

Two properties make the construction useful.
Every choice of $s$ must give a valid metric, and the rankings
must reveal none of these choices.
The Fibonacci recurrence ensures the first property:
the difference between the two possible distances across
a split equals the minimum distance between the halves
at the preceding split.

\begin{lemma}
\label{lem:lb-metric}
Every $d_s$ is a metric.
\end{lemma}

\begin{proof}
Symmetry and positivity follow from the construction.
We prove the triangle inequality inductively,
starting with pairs.

Consider a block of size $2^\ell$, $\ell\ge2$,
with halves $A,B$.
Triangles contained in either half satisfy the inequality
by induction. For a remaining triangle, assume
$u,v\in A$ and $w\in B$.
Both cross distances belong to $\{F_\ell,F_{\ell+1}\}$.
If they differ, $u$ and $v$ lie in different halves of $A$, so
\[
d_s(u,v)\ge F_{\ell-1}
=F_{\ell+1}-F_\ell
=\lvert d_s(u,w)-d_s(v,w)\rvert.
\]
If they are equal, the same inequality holds immediately.
This proves the two triangle inequalities with a cross
distance on the left.

Finally, distances within $A$ are at most $F_\ell$,
while both cross distances are at least $F_\ell$. Hence
\[
d_s(u,v)\le F_\ell
\le d_s(u,w)+d_s(w,v).
\] \qed
\end{proof}

The rankings remain unchanged because switching $s_C$
only changes whether certain comparisons are strict or tied.
For example, a vertex in $A_0$ can always rank $B_0$
before $B_1$: the former is closer when $s_C=0$,
and both are equally far when $s_C=1$.
This observation applies at every split.

\begin{lemma}
\label{lem:lb-profile}
All metrics $d_s$ are compatible with one strict ranking
profile $\pi$, with ties resolved consistently.
\end{lemma}

\begin{proof}
Fix a total order of unordered edges as follows.
First order them by the size of the smallest block containing
their endpoints, from smallest to largest.
Among edges assigned at the same block size, place those
joining corresponding quarters
($A_0$ to $B_0$ or $A_1$ to $B_1$) first.
Break remaining ties lexicographically by endpoint labels.
For pair edges, the quarter comparison is omitted.
Restrict this order to the edges incident to each vertex
to obtain $\pi$.

To verify compatibility, consider a block of size $2^\ell$
with halves $A,B$.
Distances within either half are at most $F_\ell$,
and distances between the halves are at least $F_\ell$.
Thus ranking vertices in one's own half first is valid.
For a vertex in $A_0$, distances to $B_0$ are no greater
than distances to $B_1$ under either choice of $s_C$.
The same holds for $A_1$ with $B_1$ placed first,
and symmetrically for vertices in $B$.
Within each of these quarters, all distances from the
fixed vertex are equal. The remaining tie-breaking is
therefore valid under every $d_s$. \qed
\end{proof}

An ordinal algorithm must consequently choose its tree
without knowing $s$.
To bound the cost of that tree, it is helpful to first
make its connections follow the hierarchy.
This can be done without increasing cost: reconnecting
vertices inside a block never costs more than an edge
leaving that block. Crucially, the same modification
works for every choice of $s$.

\begin{lemma}
\label{lem:lb-hierarchy}
For every spanning tree $T$, there is a spanning tree $T'$
whose induced subgraph on every block is connected and such that
\[
\cost_{d_s}(T')\le\cost_{d_s}(T)
\qquad\text{for every }s.
\]
The transformation does not depend on $s$.
\end{lemma}

\begin{proof}
Process the blocks from smallest to largest.
If the subgraph induced by a block $C$ is disconnected,
choose $u,v\in C$ in different components and add $(u,v)$.
The old path between them leaves $C$, so the resulting cycle
contains an edge $(x,y)$ with exactly one endpoint in $C$.
Remove that edge.

If $C$ has size $2^\ell$, every internal edge has length
at most $F_{\ell+1}$, while every edge leaving $C$
has length at least $F_{\ell+1}$. Thus
\[
d_s(u,v)\le F_{\ell+1}\le d_s(x,y)
\qquad\text{for every }s.
\]
The exchange reduces the number of components induced
on $C$ without increasing cost.
The removed edge is internal to no smaller block within
$C$ and no block disjoint from $C$, so previously established
connectivity is preserved.

Repeating the exchange until each block is connected
produces $T'$. Every choice in this procedure can be made
using the tree and vertex labels alone. \qed
\end{proof}

A tree connecting every block internally has exactly one
edge between the two halves of each block.
Its cost can therefore be calculated one split at a time.
An optimal tree uses an edge of length $F_\ell$ at each
block of size $2^\ell$. For a fixed tree, however,
we can choose $s_C$ so that its connecting edge has length
$F_{\ell+1}$ instead whenever $\ell\ge2$.
Pair edges always have length $1=F_1=F_2$.
There are $2^{h-\ell}$ blocks of size $2^\ell$,
giving explicit bounds on the optimal cost and the cost
forced on any fixed spanning tree.

\begin{lemma}
\label{lem:lb-costs}
For every choice $s$,
\[
\cost_{d_s}(\MST_{G_{d_s}})
=\sum_{\ell=1}^h2^{h-\ell}F_\ell
=2^{h+1}-F_{h+3}.
\]
For every spanning tree $T$, there is a choice $s$ such that
\[
\cost_{d_s}(T)\ge
\sum_{\ell=1}^h2^{h-\ell}F_{\ell+1}
=3\cdot2^h-F_{h+4}.
\]
\end{lemma}

\begin{proof}
By Lemma~\ref{lem:lb-hierarchy}, some MST connects every
block internally. Its unique edge between the halves
of a block of size $2^\ell$ costs at least $F_\ell$.
Choosing an edge of length $F_\ell$ at every split
recursively gives a spanning tree attaining all these bounds.
Hence
\[
\cost_{d_s}(\MST_{G_{d_s}})
=\sum_{\ell=1}^h2^{h-\ell}F_\ell.
\]

For an arbitrary tree $T$, obtain $T'$ from
Lemma~\ref{lem:lb-hierarchy}.
At each block $C$ of size $2^\ell$, $\ell\ge2$,
let $(u,v)$ be its unique connecting edge,
with $u\in A$ and $v\in B$.
Choose $s_C=1$ if $u\in A_0$, and $s_C=0$ if $u\in A_1$.
This gives $(u,v)$ length $F_{\ell+1}$.
Making these choices separately at every block, and
using that pair edges have length $1=F_2$, yields
\[
\cost_{d_s}(T)
\ge\cost_{d_s}(T')
=\sum_{\ell=1}^h2^{h-\ell}F_{\ell+1}.
\]

Both sum identities follow by induction on $h$.
For $h=1$, they equal $1$, as do
$2^2-F_4$ and $3\cdot2-F_5$.
For the induction step,
\[
\begin{aligned}
\sum_{\ell=1}^h2^{h-\ell}F_\ell
&=2\sum_{\ell=1}^{h-1}2^{h-1-\ell}F_\ell+F_h\\
&=2(2^h-F_{h+2})+F_h\\
&=2^{h+1}-F_{h+3},
\end{aligned}
\]
where $2F_{h+2}-F_h=F_{h+3}$ by the Fibonacci recurrence.
Similarly,
\[
\begin{aligned}
\sum_{\ell=1}^h2^{h-\ell}F_{\ell+1}
&=2\sum_{\ell=1}^{h-1}2^{h-1-\ell}F_{\ell+1}+F_{h+1}\\
&=2(3\cdot2^{h-1}-F_{h+3})+F_{h+1}\\
&=3\cdot2^h-F_{h+4},
\end{aligned}
\]
using $2F_{h+3}-F_{h+1}=F_{h+4}$. \qed
\end{proof}

The optimal cost is the same for every choice of $s$,
while any fixed output tree incurs at least the larger
cost under some choice.
The leading terms of these costs are $2^{h+1}$ and
$3\cdot2^h$, respectively. Since Fibonacci numbers grow
more slowly than powers of $2$, their ratio approaches
$3/2$. We now complete the proof of
Theorem~\ref{thm:mst-lower}.
\Needspace{9\baselineskip}
\mstlower*

\begin{proof}
Fix a deterministic ordinal algorithm $\mathcal A$.
For each $h\ge2$, let $\pi$ be the common profile from
Lemma~\ref{lem:lb-profile} and let $T=\mathcal A(\pi)$.
All metrics in the construction are compatible with
this profile.\footnote{The construction allows distance
ties, resolved by the fixed edge order from
Lemma~\ref{lem:lb-profile}. The lower bound also holds
with distinct edge lengths. Let $q=\binom n2$ and let
$r(e)\in\{1,\ldots,q\}$ be the position of edge $e$
in that order. Set $c(u,u)=0$ and
$c(u,v)=q+r(\{u,v\})$ for $u\ne v$.
Its nonzero distances lie between $q+1$ and $2q$,
so $c$ is a metric. For $0<\eta<1/(2q)$,
the metrics $d_s+\eta c$ induce the same strict
profile $\pi$: original unequal distances differ by
at least $1$, and ties are resolved by $r$.
As $\eta\to0$, tree costs and optimal costs converge
to their unperturbed values. Thus any strict lower
bound on the unperturbed ratio is preserved for
sufficiently small $\eta$.}

By Lemma~\ref{lem:lb-costs}, some choice $s$ satisfies
\begin{equation}
\frac{\cost_{d_s}(T)}
     {\cost_{d_s}(\MST_{G_{d_s}})}
\ge
\frac{3\cdot2^h-F_{h+4}}
     {2^{h+1}-F_{h+3}}
=
\frac32-
\frac{F_h}{2(2^{h+1}-F_{h+3})}.
\label{eq:lb-finite}
\end{equation}
The equality uses
$2F_{h+4}-3F_{h+3}=F_h$.

Let $\varphi=(1+\sqrt5)/2<2$.
Since $F_h=\Theta(\varphi^h)$ and
$2^{h+1}-F_{h+3}=\Theta(2^h)$, the final term
in \eqref{eq:lb-finite} tends to zero.
For every $\varepsilon>0$, choosing $h$ sufficiently
large therefore gives a ratio strictly greater than
$3/2-\varepsilon$.
Thus no deterministic ordinal algorithm guarantees
distortion at most $3/2-\varepsilon$ on all finite
metric instances. \qed
\end{proof}

Table~\ref{tab:mst-lower-bound-values} illustrates the
finite-size bound in \eqref{eq:lb-finite}.

\begin{table}[htbp]
\centering
\caption{Optimal cost, the cost forced on any fixed output tree
by a suitable choice of metric, and the resulting lower bound.}
\label{tab:mst-lower-bound-values}
\begingroup
\renewcommand{\arraystretch}{1.4}
\begin{tabular*}{0.95\linewidth}{
    @{\extracolsep{\fill}}lccccc@{}
}
\hline
Number of vertices & $4$ & $8$ & $16$ & $32$ & $64$ \\
\hline
Optimal cost & $3$ & $8$ & $19$ & $43$ & $94$ \\
Forced cost  & $4$ & $11$ & $27$ & $62$ & $137$ \\
Lower bound & $\frac43$ & $\frac{11}{8}$
            & $\frac{27}{19}$ & $\frac{62}{43}$
            & $\frac{137}{94}$ \\[3pt]
\hline
\end{tabular*}
\endgroup
\end{table}

%% file: Files/MST-general.tex
\section{Upper Bound}

We give an ordinal algorithm for constructing a spanning tree with logarithmic distortion in general metrics. The algorithm follows a Bor\r{u}vka-type template. It maintains a forest, repeatedly lets each current component propose one ordinally best outgoing edge, adds an acyclic subset of these proposals, and contracts the resulting connected components. 

Algorithm~\ref{alg:mst-general} presents our algorithm. The input is an ordinal
profile $\Profile$ over the vertex set $\Vertices$. The algorithm incrementally
constructs a forest $F$ and maintains the collection $\mathcal{C}$ of its
connected components. Initially, $F=\emptyset$ and $\mathcal{C}$ consists of the
$\NumVertices$ singleton components $\{\firstV\}$ for all
$\firstV\in\Vertices$.

In each round, every component $\Component\in\mathcal{C}$ chooses an arbitrary
representative $r(\Component)\in\Component$. The representative uses its
ordinal ranking to identify its most preferred vertex outside $\Component$,
namely $\GB_{\Component}(r(\Component))$, and the component proposes the
corresponding edge. Once every component has made a proposal, the algorithm
selects a maximal acyclic subset of the proposed edges and adds them to $F$.
Each selected edge connects two distinct components, thereby contracting them into a
single larger component. The process repeats until only one connected component
remains, at which point $F$ is a spanning tree.

\begin{algorithm}[H]
\caption{Ordinal \MST algorithm for general metrics}
\label{alg:mst-general}
\begin{algorithmic}[1]
\Require Ordinal profile $\Profile=(\pi_{\firstV})_{\firstV\in\Vertices}$
\Ensure A spanning tree $F$ on $\Vertices$
\State $F\gets \emptyset$
\State $\mathcal{C}\gets \bigl\{\{\firstV\}:\firstV\in\Vertices\bigr\}$
\While{$|\mathcal{C}|>1$}
    \State $E\gets \emptyset$
    \For{each component $\Component\in\mathcal{C}$}
        \State choose an arbitrary representative $r(\Component)\in\Component$
        \State $g_{\Component}\gets \GB_{\Component}(r(\Component))$
        \State $E\gets E\cup\bigl\{\{r(\Component),g_{\Component}\}\bigr\}$
    \EndFor
    \State choose a maximal subset $A\subseteq E$ such that $F\cup A$ is acyclic
    \State $F\gets F\cup A$
    \State $\mathcal{C}\gets$ connected components of $(\Vertices,F)$
\EndWhile
\State \Return $F$
\end{algorithmic}
\end{algorithm}

For the analysis, fix an ordinal profile $\Profile$ and a metric
$\Metric\in\ConsMetrics{\Profile}$. With respect to this fixed metric, we write $\OPT=\cost_{\Metric}(\MST_G)$.

The distortion analysis reduces to a per-round bound. In every round, the
algorithm adds a maximal acyclic subset of the proposed edges. Lemma~\ref{lem:mst-round}
shows that the cost added in any such round is at most $2\OPT$, and that the
same round reduces the number of current components by at least a factor of two.
Thus the total cost of the algorithm is obtained by summing this bound over the
rounds.

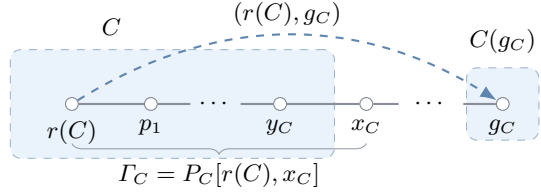
\begin{figure}[t]
\centering
\begin{tikzpicture}[
scale=0.95,
every node/.style={font=\small},
vertex/.style={circle,draw=figgray!85,fill=white,inner sep=1.7pt},
mst/.style={thick,draw=figgray!80},
ord/.style={thick,dashed,draw=figblue!90,-{Latex[length=2.2mm]}},
cluster/.style={draw=figblue!55,fill=figlightblue,dashed,rounded corners,inner sep=11pt}
]
\node[vertex,label=below:{$r(\Component)$}] (r) at (0,0) {};
\node[vertex,label=below:{$p_1$}] (p1) at (1.1,0) {};
\node at (2.0,0) {$\cdots$};
\node[vertex,label=below:{$y_{\Component}$}] (y) at (2.9,0) {};
\node[vertex,label=below:{$x_{\Component}$}] (x) at (4.1,0) {};
\node at (5.0,0) {$\cdots$};
\node[vertex,label=below:{$g_{\Component}$}] (g) at (6.0,0) {};

% Extra invisible points make the component box slightly wider and clearer.
\coordinate (cL) at (-0.45,0.35);
\coordinate (cR) at (3.25,-0.35);

% Component boxes are drawn in the background so they do not cover vertices.
\begin{pgfonlayer}{background}
\node[cluster,fit=(cL)(cR)(r)(p1)(y)] (Cbox) {};
\node[cluster,fit=(g)] (Gbox) {};
\end{pgfonlayer}

\node[above=2pt] at ([xshift=-0.85cm]Cbox.north) {$\Component$};
\node[above=2pt] at (Gbox.north) {$\Component(g_{\Component})$};

\draw[mst] (r)--(p1);
\draw[mst] (p1)--(1.65,0);
\draw[mst] (2.35,0)--(y);

\draw[mst] (y)--(x);

\draw[mst] (x)--(4.55,0);
\draw[mst] (5.45,0)--(g);

\draw[ord,bend left=32] (r) to
    node[above,text=black] {$(r(\Component),g_{\Component})$} (g);

\draw[decorate,decoration={brace,mirror,amplitude=4pt},draw=figgray!70]
    (0,-0.55) -- (4.1,-0.55);
\node[below=0.15cm] at (2.05,-0.55)
    {$\Gamma_{\Component}=P_{\Component}[r(\Component),x_{\Component}]$};

\end{tikzpicture}
\caption{Charging an ordinal edge to a prefix of the corresponding $\MST_G$-path. The path from $r(\Component)$ to $g_{\Component}$ first leaves $\Component$ at $x_{\Component}$, and the algorithmic edge $(r(\Component),g_{\Component})$ is charged to the prefix $\Gamma_{\Component}=P_{\Component}[r(\Component),x_{\Component}]$.}
\label{fig:mst-charging}
\end{figure}

\begin{lemma}
\label{lem:mst-round}
Consider any round that starts with components $\mathcal{C}=\{\Component_1,\ldots,\Component_k\}$. The total cost of the edges added in the round is at most $2\OPT$. Moreover, after the round, the number of components is at most $k/2$.
\end{lemma}

\begin{proof}
Let $A$ be the maximal acyclic subset of proposed edges added in the round. We
first bound its cost.

Consider an edge
$(r(\Component),g_{\Component})\in A$ proposed by component $\Component$,
where $g_{\Component}=\GB_{\Component}(r(\Component))$.

Let $P_{\Component}$ be the unique path in $\MST_G$ from
$r(\Component)$ to $g_{\Component}$. Since $r(\Component)\in\Component$ and
$g_{\Component}\notin\Component$, this path must leave $\Component$. Let
$x_{\Component}$ be the first vertex of $P_{\Component}$ outside $\Component$
when walking from $r(\Component)$ toward $g_{\Component}$, as in
Figure~\ref{fig:mst-charging}.

Because $g_{\Component}$ is the highest-ranked vertex outside $\Component$ for
$r(\Component)$, and $x_{\Component}$ is also outside $\Component$, consistency
implies
\[
\Metric(r(\Component),g_{\Component})
\le
\Metric(r(\Component),x_{\Component}).
\]
By the triangle inequality along the prefix of $P_{\Component}$ from
$r(\Component)$ to $x_{\Component}$,
\[
    \Metric(r(\Component),x_{\Component})
    \le
    \sum_{e\in P_{\Component}[r(\Component),x_{\Component}]} \Metric(e).
\]
Every edge in this prefix has at least one endpoint in $\Component$: all
vertices before $x_{\Component}$ lie in $\Component$, and the final edge of the
prefix crosses from $\Component$ to $\Vertices\setminus\Component$. Hence
\[
    \Metric(r(\Component),g_{\Component})
    \le
    \sum_{e\in \MST_G:\, e\cap\Component\neq\emptyset} \Metric(e).
\]
Summing over all components whose proposed edge is added gives
\[
    \cost_{\Metric}(A)
    \le
    \sum_{\Component\in\mathcal{C}}
    \sum_{e\in \MST_G:\, e\cap\Component\neq\emptyset}\Metric(e).
\]
The components of $\mathcal{C}$ are disjoint, so every \MST edge has endpoints in
at most two components and appears on the right-hand side at most twice.
Therefore
\[
    \cost_{\Metric}(A)
    \le
    2\cost_{\Metric}(\MST_G)
    =
    2\OPT.
\]

It remains to prove the contraction bound. Consider the graph on the old
components $\mathcal{C}$ whose edges are the selected proposed edges in $A$.
Suppose that some old component $\Component$ is isolated in this graph after the
round. Then no selected edge is incident to $\Component$. However,
$\Component$ proposed an edge to another component. Adding this proposed edge to
the selected component graph cannot create a cycle, because $\Component$ is
isolated. Equivalently, adding the corresponding edge to $F\cup A$ would keep
the graph acyclic. This contradicts the maximality of $A$. Thus every new
connected component contains at least two old components, and the number of
components after the round is at most $k/2$. \qed
\end{proof}

Lemma~\ref{lem:mst-round} is the key ingredient in the distortion analysis. It shows
that every round makes two kinds of progress at once: the added cost is bounded
by a constant-factor copy of the optimum, and the number of current components
shrinks by a constant factor. Since the algorithm begins with singleton
components and ends when a single component remains, this per-round progress can
occur only logarithmically many times. This gives the desired logarithmic
distortion bound.

\mstgeneral*

\begin{proof}
Algorithm~\ref{alg:mst-general} adds only acyclic subsets of proposed edges.
Therefore $F$ remains acyclic throughout the execution. The algorithm terminates
only when all vertices belong to one connected component, so the final forest is
connected. Hence the output is a spanning tree.

By Lemma~\ref{lem:mst-round}, every round costs at most $2\OPT$ and reduces the
number of connected components by at least a factor of two. Since the algorithm
starts with $\NumVertices$ singleton components, there are at most
$\lceil\log_2 \NumVertices\rceil$ rounds. Summing the per-round cost bound gives
\[
    \cost_{\Metric}(F)
    \le
    2\lceil \log_2 \NumVertices\rceil \OPT.
\]
Since this bound holds for every metric $\Metric\in\ConsMetrics{\Profile}$ and
every ordinal profile $\Profile$, the distortion of Algorithm~\ref{alg:mst-general}
is $O(\log \NumVertices)$. \qed
\end{proof}

%% file: Files/MST-tree.tex
\section{MST on Tree Metrics}

The results for general metrics leave a substantial gap: our lower bound is only
$3/2$, while the algorithmic guarantee is logarithmic. This motivates looking
for structured metric classes where the ordinal information may be more
informative. The most immediate restriction is a line metric, but this is a very
special case of a tree metric. We therefore study tree metrics as the next
natural class: they strictly generalize line metrics, while still having enough
combinatorial structure to make a sharper analysis possible. Tree metrics are
also a standard and central class in metric algorithms and metric embedding
theory~\cite{bartal1996probabilistic,fakcharoenphol2003tight}.

A metric $\Metric$ on $\Vertices$ is a \emph{tree metric} if there exists a
nonnegatively weighted tree $\RealTree=(\Vertices,\RealTreeEdges,w)$ such that
$\Metric(\firstV,\secondV)$ is the length of the unique
$\firstV$--$\secondV$ path in this tree for every
$\firstV,\secondV\in\Vertices$. We refer to this tree as a realization tree for
$\Metric$.

We now assume that the unknown metric $\Metric$ is a tree metric, induced by a
realization tree $\RealTree=(\Vertices,\RealTreeEdges,w)$. The algorithm is the
same round-based algorithm as Algorithm~\ref{alg:mst-general}, except that each
component redirects its proposed edge before the maximal acyclic subset is
chosen. Specifically, when a component $\Component$ chooses representative
$r(\Component)\in\Component$ and globally best outside vertex
\[
    g_{\Component}=\GB_{\Component}(r(\Component)),
\]
the algorithm does not propose the edge
$(r(\Component),g_{\Component})$. Instead, it redirects from
$g_{\Component}$ to its locally best vertex inside $\Component$:
\[
    h_{\Component}=\LB_{\Component}(g_{\Component}),
\]
and proposes the edge $(h_{\Component},g_{\Component})$. The algorithm then
selects a maximal subset of these redirected proposed edges whose addition keeps
the forest acyclic, adds this subset to $F$, and updates the connected
components.

For the analysis, fix a round, a component $\Component$, its representative
$r(\Component)$, and its globally best outside vertex
$g_{\Component}=\GB_{\Component}(r(\Component))$. Let $P_{\Component}$ be the
path in the realization tree $\RealTree$ from $r(\Component)$ to
$g_{\Component}$. Let $x_{\Component}$ be the first vertex of $P_{\Component}$
outside $\Component$ when walking from $r(\Component)$ to $g_{\Component}$, and
let $y_{\Component}$ be the vertex immediately before $x_{\Component}$ on this
path. Thus $y_{\Component}\in\Component$, $x_{\Component}\notin\Component$, and
$(y_{\Component},x_{\Component})\in\RealTreeEdges$. This local configuration
is illustrated in Figure~\ref{fig:tree-local-charge}.

Since $x_{\Component}$ is outside $\Component$ and $g_{\Component}$ is the
globally best outside vertex for $r(\Component)$, consistency gives
$\Metric(r(\Component),g_{\Component})\le
\Metric(r(\Component),x_{\Component})$. On the other hand, $x_{\Component}$
lies on the tree path from $r(\Component)$ to $g_{\Component}$, so
\[
\Metric(r(\Component),g_{\Component})
=
\Metric(r(\Component),x_{\Component})
+
\Metric(x_{\Component},g_{\Component}).
\]
By the definition of the globally best outside vertex and this equality, the distance between $x_{\Component}$ and $g_{\Component}$ is zero.
\begin{observation}
\label{obs:tree-zero-suffix}
For the fixed round and component $\Component$,
\[
    \Metric(x_{\Component},g_{\Component})=0.
\]
\end{observation}

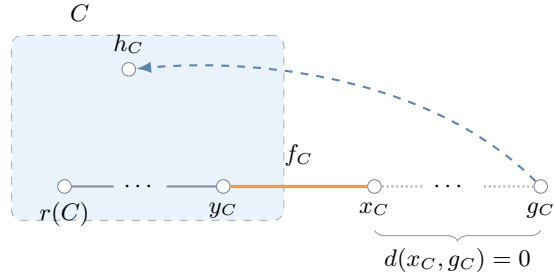
\begin{figure}[H]
\centering
\begin{tikzpicture}[
x=1cm,y=1cm,
every node/.style={font=\small},
vertex/.style={circle,draw=figgray!85,fill=white,inner sep=1.8pt},
treeedge/.style={thick,draw=figgray!80},
posedge/.style={line width=1.2pt,draw=figorange!90},
zeroedge/.style={thick,densely dotted,draw=figgray!55},
ord/.style={thick,dashed,draw=figblue!90,-{Latex[length=2mm]}},
cluster/.style={draw=figblue!55,fill=figlightblue,dashed,rounded corners,inner sep=10pt}
]
\node[vertex,label=below:{$r(\Component)$}] (r) at (0,0) {};
\node at (1.0,0) {$\cdots$};
\node[vertex,label=below:{$y_{\Component}$}] (y) at (2.1,0) {};
\node[vertex,label=below:{$x_{\Component}$}] (x) at (4.1,0) {};
\node at (5.1,0) {$\cdots$};
\node[vertex,label=below:{$g_{\Component}$}] (g) at (6.3,0) {};
\node[vertex,label=above:{$h_{\Component}$}] (h) at (0.85,1.55) {};

% Component C.
\coordinate (cL) at (-0.35,0);
\coordinate (cR) at (2.55,0);
\coordinate (cT) at (0.85,1.55);

% Component box is drawn in the background so it does not cover vertices.
\begin{pgfonlayer}{background}
\node[cluster,fit=(cL)(cR)(cT)(r)(y)(h)] (Cbox) {};
\end{pgfonlayer}

\node[above=2pt] at ([xshift=-0.90cm]Cbox.north) {$\Component$};

% Tree path from r(C) to g_C.
\draw[treeedge] (r)--(0.65,0);
\draw[treeedge] (1.35,0)--(y);

\draw[posedge] (y)--node[above=4pt,text=black] {$f_{\Component}$} (x);

\draw[zeroedge] (x)--(4.75,0);
\draw[zeroedge] (5.55,0)--(g);

% Zero-length suffix label.
\draw[decorate,decoration={brace,mirror,amplitude=4pt},draw=figgray!70]
    (4.1,-0.58) -- (6.3,-0.58);
\node[below=0.12cm] at (5.2,-0.58)
    {$\Metric(x_{\Component},g_{\Component})=0$};

% Redirected ordinal edge, kept away from the boundary of C.
\draw[ord] (g) to[out=135,in=5,looseness=0.85] (h);

\end{tikzpicture}
\caption{Local tree-metric charging. The path from $r(\Component)$ to $g_{\Component}$ first leaves $\Component$ through the positive tree edge $f_{\Component}=(y_{\Component},x_{\Component})$. Since $g_{\Component}$ is globally best outside $\Component$, the suffix from $x_{\Component}$ to $g_{\Component}$ has zero length. Therefore the redirected edge to $h_{\Component}$ is no more expensive than $f_{\Component}$, unless it already has zero cost.}
\label{fig:tree-local-charge}
\end{figure}

Observation~\ref{obs:tree-zero-suffix} explains why the redirection step is
useful. Once the path from $r(\Component)$ to $g_{\Component}$ leaves
$\Component$ at $x_{\Component}$, the remaining suffix to $g_{\Component}$ has
zero length. Since the redirected endpoint $h_{\Component}$ is the locally best
vertex of $g_{\Component}$ inside $\Component$, the cost of the redirected edge
can be charged to the first realization-tree edge by which the path leaves
$\Component$.

\begin{lemma}
\label{lem:tree-local-charge}
For every proposed edge $(g_{\Component},h_{\Component})$ in the tree-metric
algorithm, if $\Metric(g_{\Component},h_{\Component})>0$, then
\[
\Metric(g_{\Component},h_{\Component})
\le
\Metric(y_{\Component},x_{\Component}),
\]
where $(y_{\Component},x_{\Component})$ is the first realization-tree edge by
which the path from $r(\Component)$ to $g_{\Component}$ leaves $\Component$.
\end{lemma}
\begin{proof}
By Observation~\ref{obs:tree-zero-suffix},
$\Metric(x_{\Component},g_{\Component})=0$. Since
$y_{\Component}\in\Component$ and
$h_{\Component}=\LB_{\Component}(g_{\Component})$ is the highest-ranked vertex
of $g_{\Component}$ inside $\Component$, consistency implies that the redirected
endpoint is no farther from $g_{\Component}$ than $y_{\Component}$ is:
$\Metric(g_{\Component},h_{\Component})\le
\Metric(g_{\Component},y_{\Component})$.

Now the path from $g_{\Component}$ to $y_{\Component}$ first follows the
zero-length suffix from $g_{\Component}$ to $x_{\Component}$ and then crosses
the single tree edge $(x_{\Component},y_{\Component})$. Hence
$\Metric(g_{\Component},y_{\Component})
=
\Metric(x_{\Component},y_{\Component})$. Combining these two facts gives
\[
\Metric(g_{\Component},h_{\Component})
\le
\Metric(x_{\Component},y_{\Component}),
\]
as claimed. \qed
\end{proof}

The remaining point is to show that the same positive realization-tree edge cannot be charged twice. We prove this by contradiction. Figure~\ref{fig:tree-double-charge} summarizes the two possible types of double charge: an earlier-round double charge and a same-round double charge.

\begin{figure}[t]
\centering
\begin{tikzpicture}[
x=0.92cm,y=0.92cm,
every node/.style={font=\small},
vertex/.style={circle,draw=figgray!85,fill=white,inner sep=1.55pt},
treeedge/.style={thick,draw=figgray!80},
posedge/.style={line width=1.2pt,draw=figorange!90},
zeroedge/.style={thick,densely dotted,draw=figgray!55},
ord/.style={thick,dashed,draw=figblue!90,-{Latex[length=1.8mm]}},
clusterC/.style={draw=figblue!55,fill=figlightblue,dashed,rounded corners},
clusterD/.style={draw=figgreen!55,fill=figlightgreen,dashed,rounded corners}
]

% ============================================================
% (a1) Earlier-round: merged component already inside C
% ============================================================
\begin{scope}[shift={(0,0)}]
\node[anchor=west] at (-0.3,1.95)
{\textbf{(a1)} Earlier-round double charge: the earlier merged component is already inside $\Component$.};

\draw[clusterC] (-0.55,-0.65) rectangle (3.15,1.45);
\node at (-0.18,1.20) {$\Component$};

\node[vertex,label=below:{$y_{\Component}$}] (aY) at (1.35,0) {};
\node[vertex,label=below:{$x_{\Component}$}] (aX) at (3.35,0) {};
\node[vertex,label=left:{$g_D$}] (aGD) at (2.35,0.78) {};
\node[vertex,label=below:{$g_{\Component}$}] (aGC) at (5.45,0) {};
\node[vertex,label=above:{$h_{\Component}$}] (aHC) at (0.45,0.78) {};

\draw[posedge] (aY)--node[below,text=black] {$f_{\Component}$} (aX);
\draw[zeroedge] (aGD)--node[above right,text=black] {$0$} (aX);
\draw[zeroedge] (aX)--node[below,text=black] {$0$} (aGC);
\draw[ord] (aGC) to[out=140,in=20] (aHC);

\node[align=left,anchor=west] at (6.10,0.55)
{$g_D\in\Component$ and $\Metric(g_D,g_{\Component})=0$};
\node[align=left,anchor=west] at (6.10,0.05)
{$\Rightarrow \Metric(g_{\Component},h_{\Component})=0$};

\end{scope}

% ============================================================
% (a2) Earlier-round: merged component outside C
% ============================================================
\begin{scope}[shift={(0,-3.25)}]
\node[anchor=west] at (-0.3,1.95)
{\textbf{(a2)} Earlier-round double charge: the earlier merged component remains outside $\Component$.};

\draw[clusterC] (-0.55,-0.65) rectangle (2.55,0.65);
\node at (-0.18,0.42) {$\Component$};

\node[vertex,label=below:{$r(\Component)$}] (bR) at (0.25,0) {};
\node at (0.95,0) {$\cdots$};
\node[vertex,label=below:{$y_{\Component}$}] (bY) at (1.85,0) {};
\node[vertex,label=below:{$x_{\Component}$}] (bX) at (3.65,0) {};
\node[vertex,label=below:{$g_{\Component}$}] (bGC) at (5.45,0) {};
\node[vertex,label=above:{$g_D$}] (bGD) at (1.85,0.95) {};

\draw[treeedge] (bR)--(0.65,0);
\draw[treeedge] (1.20,0)--(bY);
\draw[posedge] (bY)--node[below,text=black] {$f_{\Component}$} (bX);
\draw[zeroedge] (bY)--node[right,text=black] {$0$} (bGD);
\draw[zeroedge] (bX)--node[below,text=black] {$0$} (bGC);

\node[align=left,anchor=west] at (6.10,0.55)
{$g_D\notin\Component$ and};
\node[align=left,anchor=west] at (6.10,0.05)
{$\Metric(r(\Component),g_D)<\Metric(r(\Component),g_{\Component})$};

\end{scope}

% ============================================================
% (b1) Same-round: g_D outside C
% ============================================================
\begin{scope}[shift={(0,-6.50)}]
\node[anchor=west] at (-0.3,1.95)
{\textbf{(b1)} Same-round double charge: $g_D\notin\Component$.};

\draw[clusterC] (-0.55,-0.65) rectangle (2.55,0.65);
\node at (-0.18,0.42) {$\Component$};

\draw[clusterD] (3.05,-0.65) rectangle (4.65,0.65);
\node at (4.38,0.42) {$D$};

\node[vertex,label=below:{$r(\Component)$}] (cR) at (0.25,0) {};
\node at (0.95,0) {$\cdots$};
\node[vertex,label=below:{$y_{\Component}$}] (cY) at (1.85,0) {};
\node[vertex,label=below:{$x_{\Component}$}] (cX) at (3.65,0) {};
\node[vertex,label=below:{$g_{\Component}$}] (cGC) at (5.45,0) {};
\node[vertex,label=above:{$g_D$}] (cGD) at (1.85,0.95) {};

\draw[treeedge] (cR)--(0.65,0);
\draw[treeedge] (1.20,0)--(cY);
\draw[posedge] (cY)--node[below,text=black] {$f_{\Component}$} (cX);
\draw[zeroedge] (cY)--node[right,text=black] {$0$} (cGD);
\draw[zeroedge] (cX)--node[below,text=black] {$0$} (cGC);

\node[align=left,anchor=west] at (6.10,0.55)
{$g_D$ is a better outside};
\node[align=left,anchor=west] at (6.10,0.05)
{candidate for $r(\Component)$};

\end{scope}

% ============================================================
% (b2) Same-round: g_D in C but g_C outside D
% ============================================================
\begin{scope}[shift={(0,-9.75)}]
\node[anchor=west] at (-0.3,1.95)
{\textbf{(b2)} Same-round double charge: $g_D\in\Component$ but $g_{\Component}\notin D$.};

\draw[clusterC] (-0.55,-0.65) rectangle (2.75,1.45);
\node at (-0.18,1.20) {$\Component$};

\draw[clusterD] (3.05,-0.65) rectangle (5.05,1.45);
\node at (4.78,1.20) {$D$};

\node[vertex,label=below:{$y_{\Component}$}] (dY) at (2.05,0) {};
\node[vertex,label=below:{$x_{\Component}$}] (dX) at (3.65,0) {};
\node[vertex,label=above:{$g_D$}] (dGD) at (1.20,0.78) {};
\node[vertex,label=above:{$r(D)$}] (dRD) at (3.85,0.78) {};
\node[vertex,label=below:{$g_{\Component}$}] (dGC) at (5.55,0) {};

\draw[posedge] (dY)--node[below,text=black] {$f_{\Component}$} (dX);
\draw[zeroedge] (dY)--node[below left=-1pt,text=black] {$0$} (dGD);
\draw[zeroedge] (dX)--node[below,text=black] {$0$} (dGC);
\draw[ord] (dRD) to[out=145,in=35] (dGD);

\node[align=left,anchor=west] at (6.10,0.55)
{$g_{\Component}$ is a better outside};
\node[align=left,anchor=west] at (6.10,0.05)
{candidate for $r(D)$};

\end{scope}

\end{tikzpicture}

\caption{Double-charging cases for a positive realization-tree edge
$f_{\Component}=(y_{\Component},x_{\Component})$. Panels (a1)--(a2) show the
earlier-round case. Panels (b1)--(b2) show the same-round case. Each row
corresponds to one subcase in the proof of
Lemma~\ref{lem:tree-no-double-charge}.}
\label{fig:tree-double-charge}
\end{figure}
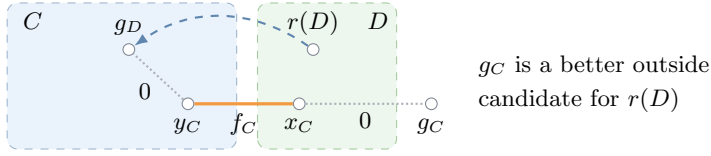
\begin{lemma}
\label{lem:tree-no-double-charge}
Across the whole execution of the tree-metric algorithm, every positive edge of
the realization tree $\RealTree$ is charged by at most one added edge.
\end{lemma}

\begin{proof}
Suppose, for contradiction, that a positive realization-tree edge is charged
twice. Among the two charges, choose one that occurs in a latest round, and let
this charge come from component $\Component$.

Let
\[
f_{\Component}=(y_{\Component},x_{\Component})
\]
be the charged edge, oriented so that $y_{\Component}\in\Component$ and
$x_{\Component}\notin\Component$. Denote $r(\Component)$ be the representative of
$\Component$, and let
\[
g_{\Component}=\GB_{\Component}(r(\Component))
\qquad\text{and}\qquad
h_{\Component}=\LB_{\Component}(g_{\Component}).
\]
By Observation~\ref{obs:tree-zero-suffix}, $\Metric(x_{\Component},g_{\Component})=0.$ Since this charge is to the positive edge $f_{\Component}$, the redirected edge
has positive cost:
\[
\Metric(g_{\Component},h_{\Component})>0.
\]

Let the other charge come from component $D$, with representative $r(D)$ and
globally best outside vertex
\[
g_D=\GB_D(r(D)).
\]

\medskip
\noindent\textbf{Case 1: the other charge occurred in an earlier round.}
After the earlier charged edge was added, its endpoints were merged into one
current component. Since that earlier edge was charged to the same tree edge
$f_{\Component}$, this merged component reaches across $f_{\Component}$ through
zero-length parts on both sides of $f_{\Component}$.

At the time of the later charge by $\Component$, this earlier merged component
is either contained in $\Component$ or disjoint from $\Component$.

If it is contained in $\Component$, then, as shown in
Figure~\ref{fig:tree-double-charge}(a1), component $\Component$ contains $g_D$,
and $g_D$ is at zero distance from $x_{\Component}$. Since
$\Metric(x_{\Component},g_{\Component})=0$, we get
$\Metric(g_D,g_{\Component})=0$. Therefore the locally best vertex
$h_{\Component}=\LB_{\Component}(g_{\Component})$ satisfies
\[
\Metric(g_{\Component},h_{\Component})=0,
\]
contradicting $\Metric(g_{\Component},h_{\Component})>0$.

If it is disjoint from $\Component$, then, as shown in
Figure~\ref{fig:tree-double-charge}(a2), $g_D$ is outside $\Component$ and is at
zero distance from $y_{\Component}$. Thus $g_D$ is an outside candidate for
$r(\Component)$. The path from $r(\Component)$ to $g_D$ reaches
$y_{\Component}$ and then has zero additional length. The path from
$r(\Component)$ to $g_{\Component}$ reaches $y_{\Component}$, crosses the
positive edge $f_{\Component}$, and then has zero additional length from
$x_{\Component}$ to $g_{\Component}$. Hence
\[
\Metric(r(\Component),g_D)
<
\Metric(r(\Component),g_{\Component}),
\]
contradicting the choice of $g_{\Component}=\GB_{\Component}(r(\Component))$.

\medskip
\noindent\textbf{Case 2: the other charge occurs in the same round.}
Since $\Component$ and $D$ are distinct current components and
$y_{\Component}\in\Component$, component $D$ cannot contain $y_{\Component}$.
Therefore, for $D$ to charge the same edge $f_{\Component}$, its path must cross
$f_{\Component}$ in the opposite direction. Thus $x_{\Component}\in D$,
$y_{\Component}\notin D$, and applying Observation~\ref{obs:tree-zero-suffix}
to component $D$ gives
\[
\Metric(y_{\Component},g_D)=0.
\]

If $g_D\notin\Component$, then $g_D$ is an outside candidate for
$r(\Component)$, as shown in Figure~\ref{fig:tree-double-charge}(b1). The path
from $r(\Component)$ to $g_D$ reaches $y_{\Component}$ and then has zero
additional length, while the path from $r(\Component)$ to $g_{\Component}$
crosses the positive edge $f_{\Component}$. Hence
\[
\Metric(r(\Component),g_D)
<
\Metric(r(\Component),g_{\Component}),
\]
contradicting the choice of $g_{\Component}$.

Thus $g_D\in\Component$. Now if $g_{\Component}\notin D$, then
$g_{\Component}$ is an outside candidate for $r(D)$, as shown in
Figure~\ref{fig:tree-double-charge}(b2). Since
$\Metric(x_{\Component},g_{\Component})=0$, the path from $r(D)$ to
$g_{\Component}$ reaches $x_{\Component}$ and then has zero additional length,
while the path from $r(D)$ to $g_D$ crosses the positive edge
$f_{\Component}$. Hence
\[
\Metric(r(D),g_{\Component})
<
\Metric(r(D),g_D),
\]
contradicting the choice of $g_D$.

The only remaining possibility is that $g_D\in\Component$ and
$g_{\Component}\in D$. Then the proposal of $\Component$ connects $\Component$
to $D$, and the proposal of $D$ also connects $D$ to $\Component$. If both
redirected edges were added, then together with the existing paths inside
$\Component$ and inside $D$, they would create a cycle. This contradicts the
fact that the algorithm adds an acyclic subset of the proposed edges.

Both cases lead to contradictions. Therefore no positive realization-tree edge
is charged twice. \qed
\end{proof}

For the final comparison with the optimum \MST, we use the simple fact that in a
tree metric the realization tree itself has exactly the optimum spanning-tree
cost. Indeed, every spanning tree in the complete graph must connect across each
realization-tree cut at least once, so it must pay for every realization-tree
edge.

\begin{lemma}
\label{lem:tree-mst-value}
For a tree metric $\Metric$ realized by
$\RealTree=(\Vertices,\RealTreeEdges,\Weight)$, we have 
\[
    \cost_{\Metric}(\MST_G)
    =
    \sum_{e\in\RealTreeEdges}\Weight(e).
\]
\end{lemma}

\begin{proof}
The realization tree $\RealTree$ itself is a feasible spanning tree in the
complete graph metric, so
\[
\cost_{\Metric}(\MST_G)
\le
\sum_{e\in\RealTreeEdges}\Weight(e).
\]

For the reverse inequality, let $T$ be any spanning tree on $\Vertices$ in the
complete graph metric. For a realization-tree edge $e\in\RealTreeEdges$, let $m_e$ be the number of
edges $(\firstV,\secondV)\in T$ whose path in $\RealTree$ contains $e$.
Equivalently, $m_e$ is the number of edges of $T$ crossing the cut created by
deleting $e$ from $\RealTree$. Since $T$ is connected, every such cut is crossed
by at least one edge of $T$, and therefore $m_e\ge 1$ for every
$e\in\RealTreeEdges$.

Thus
\[
\cost_{\Metric}(T)
=
\sum_{(\firstV,\secondV)\in T}
\sum_{e\in P_{\RealTree}(\firstV,\secondV)} \Weight(e)
=
\sum_{e\in\RealTreeEdges} m_e\Weight(e)
\ge
\sum_{e\in\RealTreeEdges}\Weight(e).
\] \qed
\end{proof}

The preceding lemmas complete the charging argument. Lemma~\ref{lem:tree-local-charge}
shows that every positive-cost edge added by the algorithm can be charged to a
positive edge of the realization tree. Lemma~\ref{lem:tree-no-double-charge}
shows that no positive realization-tree edge receives more than one charge.
Finally, Lemma~\ref{lem:tree-mst-value} shows that the optimum \MST cost is
exactly the sum of the weights of all realization-tree edges. Combining these
facts gives distortion one.

\msttree*

\begin{proof}
The algorithm adds only subsets of proposed edges whose addition keeps the
current forest acyclic. Therefore $F$ remains acyclic throughout the execution.
The algorithm terminates only when all vertices belong to one connected
component, so the final forest is connected. Hence the output is a spanning
tree.

By Lemma~\ref{lem:tree-local-charge}, every added edge has either zero cost or
can be charged to one positive edge of the realization tree. By
Lemma~\ref{lem:tree-no-double-charge}, no positive realization-tree edge is
charged more than once. Therefore
\[
    \cost_{\Metric}(F)
    \le
    \sum_{e\in\RealTreeEdges}\Weight(e)
    =
    \cost_{\Metric}(\MST_G),
\]
where the equality follows from Lemma~\ref{lem:tree-mst-value}. Since $F$ is
itself a spanning tree, its cost is also at least $\cost_{\Metric}(\MST_G)$.
Thus equality holds. \qed 
\end{proof}

%% file: Files/MST-doubling.tex
\section{MST on Doubling Metrics}
\label{sec:mst-doubling}

Theorem~\ref{thm:mst-tree} shows that an optimal minimum
spanning tree can be recovered from ordinal information
on tree metrics. We next consider metrics of bounded
doubling dimension. This class includes fixed-dimensional
Euclidean metrics and is incomparable with tree metrics:
tree metrics can have arbitrarily large doubling dimension,
while bounded-doubling metrics need not admit a tree
representation. 
We show that their geometric structure allows distortion
independent of $\NumVertices$. A metric has doubling dimension
$\Dim$ if every ball can be covered by at most $2^{\Dim}$ balls
of half its radius, centered at vertices, and $\Dim$ is the
smallest value with this property.
Thus, bounded doubling dimension limits how many smaller
balls are needed to cover any larger ball.

Our approach uses a spanning-tree property called
\emph{endpoint cut stability}, which requires that,
when removing any tree edge splits the vertices into
two sets, each endpoint is at least as close to the
other endpoint as to any other vertex in the opposite set.
More precisely, for an edge $(u,v)\in T$, let
$\Vertices_u$ and $\Vertices_v$ be the vertex sets of
the components of $T-(u,v)$ containing $u$ and $v$,
respectively. The condition is
\[
\Metric(u,v)\le\Metric(u,x)
\quad\text{for every }x\in\Vertices_v,
\qquad
\Metric(u,v)\le\Metric(v,y)
\quad\text{for every }y\in\Vertices_u.
\]
Figure~\ref{fig:endpoint-cut-stability} illustrates
these comparisons.

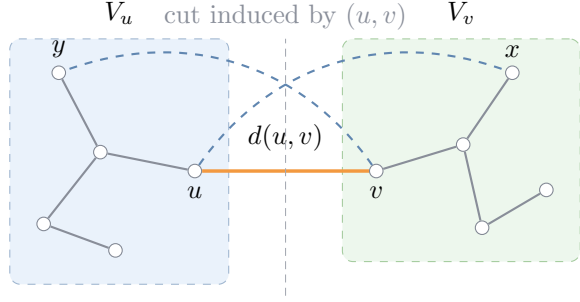
\begin{figure}[t]
\centering
\begin{tikzpicture}[
x=1cm,y=1cm,
every node/.style={font=\normalsize},
vertex/.style={
    circle,
    draw=figgray!85,
    fill=white,
    inner sep=1.8pt
},
treeedge/.style={
    thick,
    draw=figgray!80
},
stableedge/.style={
    line width=1.3pt,
    draw=figorange!90
},
comparison/.style={
    thick,
    dashed,
    draw=figblue!90
},
clusterL/.style={
    draw=figblue!55,
    fill=figlightblue,
    dashed,
    rounded corners,
    inner sep=10pt
},
clusterR/.style={
    draw=figgreen!55,
    fill=figlightgreen,
    dashed,
    rounded corners,
    inner sep=10pt
},
cutline/.style={
    densely dashed,
    draw=figgray!60
}
]
\node[vertex,label=below:{$u$}] (i) at (2.6,0) {};
\node[vertex,label=below:{$v$}] (j) at (5.0,0) {};

\node[vertex,label=above:{$y$}] (kL) at (0.8,1.3) {};
\node[vertex] (l1) at (0.6,-0.7) {};
\node[vertex] (l2) at (1.55,-1.05) {};
\node[vertex] (l3) at (1.35,0.25) {};

\node[vertex,label=above:{$x$}] (kR) at (6.8,1.3) {};
\node[vertex] (r1) at (6.4,-0.75) {};
\node[vertex] (r2) at (7.25,-0.25) {};
\node[vertex] (r3) at (6.15,0.35) {};

\draw[treeedge] (kL)--(l3);
\draw[treeedge] (l3)--(i);
\draw[treeedge] (l3)--(l1);
\draw[treeedge] (l1)--(l2);

\draw[treeedge] (j)--(r3);
\draw[treeedge] (r3)--(kR);
\draw[treeedge] (r3)--(r1);
\draw[treeedge] (r1)--(r2);

\draw[stableedge]
    (i)--node[above=4pt,text=black]
    {$\Metric(u,v)$}
    (j);

\begin{pgfonlayer}{background}
\node[
    clusterL,
    fit=(kL)(l1)(l2)(l3)(i)
] (LeftCluster) {};
\node[
    clusterR,
    fit=(j)(kR)(r1)(r2)(r3)
] (RightCluster) {};
\end{pgfonlayer}

\node[above=2pt] at (LeftCluster.north)
    {$\Vertices_u$};
\node[above=2pt] at (RightCluster.north)
    {$\Vertices_v$};

\draw[cutline]
    (3.8,-1.65) -- (3.8,1.75);
\node[above,text=figgray!80] at (3.8,1.75)
    {cut induced by $(u,v)$};

\draw[comparison]
    (i) to[out=55,in=160] (kR);
\draw[comparison]
    (j) to[out=125,in=20] (kL);
\end{tikzpicture}
\caption{Removing $(u,v)$ separates the tree into
$\Vertices_u$ and $\Vertices_v$.
Endpoint cut stability requires
$\Metric(u,v)\le\Metric(u,x)$ for every $x\in\Vertices_v$
and $\Metric(u,v)\le\Metric(v,y)$ for every $y\in\Vertices_u$.
The dashed blue curves illustrate these comparisons;
they are not tree edges.}
\label{fig:endpoint-cut-stability}
\end{figure}

We first show how to construct such a tree from the
ordinal profile $\Profile$. The algorithm follows
Kruskal's procedure, adding an edge whenever its endpoints
belong to different components. To determine the processing
order without numerical distances, we use the comparisons
provided by the rankings.

Form a directed graph whose vertices are the unordered
edges of the complete graph on $\Vertices$.
Add an arc from $(u,v)$ to $(u,w)$ whenever 
$
\rank_u(v)<\rank_u(w).
$
Every such arc implies
$\Metric(u,v)\le\Metric(u,w)$ for every
$\Metric\in\ConsMetrics{\Profile}$.
The same is true along any directed path.
Consequently, edges in the same strongly connected
component must have equal length under every consistent
metric.

Contract these strongly connected components and order
the resulting directed acyclic graph topologically.
Process its components in that order, using a fixed
order of endpoint labels within each component.
Write $\le^\ast$ for the resulting total order of edges.
Although comparisons within a strongly connected
component need not all be preserved, the edges involved
are necessarily equal in length. Comparisons between
different components are preserved.

\begin{algorithm}[t]
\caption{Ordinal construction of an endpoint cut-stable tree}
\label{alg:mst-doubling}
\begin{algorithmic}[1]
\Require Ordinal profile $\Profile$ on vertex set $\Vertices$
\State Form the directed graph of edge comparisons induced by $\Profile$
\State Compute its strongly connected components
\State Order the components topologically, breaking ties by fixed labels
\State Obtain $\le^\ast$ by ordering edges within each component by fixed labels
\State $F\gets\emptyset$; initialize a union-find structure on $\Vertices$
\ForAll{edges $(u,v)$ in increasing order of $\le^\ast$}
    \If{$u$ and $v$ lie in different components of $F$}
        \State Add $(u,v)$ to $F$ and merge their components
    \EndIf
\EndFor
\State \Return $T=(\Vertices,F)$
\end{algorithmic}
\end{algorithm}

The algorithm uses only the rankings and fixed labels,
so it is deterministic and ordinal. Processing all edges
of the complete graph produces a spanning tree.
The next lemma establishes the property needed for
its analysis.

\begin{lemma}
\label{lem:ecs-construction}
The tree returned by Algorithm~\ref{alg:mst-doubling}
is endpoint cut-stable with respect to every
$\Metric\in\ConsMetrics{\Profile}$.
\end{lemma}

\begin{proof}
Fix a consistent metric $\Metric$ and an edge $(u,v)\in T$.
Suppose that some $x\in\Vertices_v$ satisfies
\[
\Metric(u,x)<\Metric(u,v).
\]
Consistency forces $u$ to rank $x$ before $v$, so the
comparison graph contains an arc from $(u,x)$ to $(u,v)$.
These two edges cannot belong to the same strongly
connected component, since their lengths differ.
Therefore, $(u,x)$ is processed before $(u,v)$.

When $(u,x)$ is processed, its endpoints are either
already connected or become connected by adding this edge.
In either case, the current forest contains a path
from $u$ to $x$ that does not use $(u,v)$.
All edges of this path remain in the final tree.
Thus $u$ and $x$ are connected in $T-(u,v)$,
contradicting $x\in\Vertices_v$.

It follows that
$\Metric(u,v)\le\Metric(u,x)$ for every $x\in\Vertices_v$.
Exchanging $u$ and $v$ proves the other required inequality. \qed
\end{proof}

We now bound the cost of an endpoint cut-stable tree
in a doubling metric. The comparison uses a light
spanner: a graph whose shortest-path distances approximate
the original metric and whose total edge weight is
bounded relative to an \MST.

Fix a metric $\Metric$ of doubling dimension $\Dim$
and $0<\epsilon<1/2$.
By Borradaile, Le, and
Wulff-Nilsen~\cite{BorradaileLeWulffNilsen2019},
there is a graph $\Spanner$ on $\Vertices$, with each
edge $(u,v)$ weighted by $\Metric(u,v)$, such that
\[
\Metric(u,v)
\le \SpannerDist(u,v)
\le (1+\epsilon)\Metric(u,v)
\qquad\text{for all }u,v\in\Vertices,
\]
where $\SpannerDist$ denotes shortest-path distance
in $\Spanner$. Its total weight satisfies
\[
\cost_{\Metric}(\Spanner)
\le
\epsilon^{-O(\Dim)}
\cost_{\Metric}(\MST_G),
\]
where $G$ is the complete graph weighted by $\Metric$.
The spanner is used only in the analysis.

Let $T$ be endpoint cut-stable with respect to $\Metric$.
We show that every positive-weight edge of $T$ can be
assigned to a distinct spanner edge of comparable length.
Zero-weight edges, if allowed by the model, contribute
nothing to the cost and can be omitted from this assignment.

\begin{lemma}
\label{lem:single-edge}
For every edge $(u,v)\in T$ of weight
$w=\Metric(u,v)>0$, there is an edge $(a,b)\in\Spanner$
with $a\in\Vertices_u$ and $b\in\Vertices_v$ such that
\[
\Metric(a,b)\ge(1-\epsilon)w,
\qquad
\Metric(u,a)\le\epsilon w,
\qquad
\Metric(v,b)\le\epsilon w.
\]
\end{lemma}

\begin{proof}
Let $P$ be a shortest $u$--$v$ path in $\Spanner$.
Its length is at most $(1+\epsilon)w$.
Because $P$ begins in $\Vertices_u$ and ends in
$\Vertices_v$, it has an edge $(a,b)$ crossing from
$\Vertices_u$ to $\Vertices_v$.
For definiteness, take $a$ to be the last vertex of
$P$ in $\Vertices_u$ and $b$ its successor.

Let $p$ and $q$ be the lengths of the portions of $P$
from $u$ to $a$ and from $b$ to $v$, respectively,
and let $\ell=\Metric(a,b)$. Then
\begin{equation}
p+\ell+q\le(1+\epsilon)w.
\label{eq:doubling-path}
\end{equation}
Endpoint cut stability and the triangle inequality give
\[
p+\ell\ge\Metric(u,b)\ge w,
\qquad
q+\ell\ge\Metric(v,a)\ge w.
\]
Subtracting these inequalities separately from
\eqref{eq:doubling-path} yields
\[
p\le\epsilon w,
\qquad
q\le\epsilon w.
\]
Hence
\[
\Metric(u,a)\le p\le\epsilon w,
\qquad
\Metric(v,b)\le q\le\epsilon w,
\]
and
\[
\ell\ge w-p\ge(1-\epsilon)w.
\] \qed
\end{proof}

For each positive-weight tree edge, choose a spanner
edge as in Lemma~\ref{lem:single-edge}.
The endpoint distance bounds ensure that two different
tree edges cannot be assigned to the same spanner edge.

\begin{lemma}
\label{lem:injective}
The assignment is injective on the positive-weight
edges of $T$.
\end{lemma}

\begin{proof}
Suppose two distinct tree edges are assigned to the
same unordered spanner edge $\{a,b\}$.
Because this spanner edge crosses both corresponding
cuts, both tree edges lie on the unique $a$--$b$ path
in $T$.

Orient both tree edges along this path and name them
$(u,v)$ and $(u',v')$ in their order from $a$ to $b$.
Thus the path visits
\[
a,\ldots,u,v,\ldots,u',v',\ldots,b,
\]
where $v=u'$ is possible.
Write $w=\Metric(u,v)>0$ and
$w'=\Metric(u',v')>0$.
The bounds from Lemma~\ref{lem:single-edge}, with
endpoints oriented toward $a$ and $b$, give
\[
\Metric(u,a)\le\epsilon w,
\qquad
\Metric(v,b)\le\epsilon w,
\]
and
\[
\Metric(u',a)\le\epsilon w',
\qquad
\Metric(v',b)\le\epsilon w'.
\]

Removing $(u,v)$ places $u'$ on the side containing $v$.
Endpoint cut stability at $(u,v)$ therefore gives
\[
w\le\Metric(u,u')
\le\Metric(u,a)+\Metric(a,u')
\le\epsilon(w+w').
\]
Similarly, removing $(u',v')$ places $v$ on the side
containing $u'$, opposite $v'$.
Endpoint cut stability at $(u',v')$ gives
\[
w'\le\Metric(v',v)
\le\Metric(v',b)+\Metric(b,v)
\le\epsilon(w+w').
\]
Adding these inequalities yields
\[
w+w'\le2\epsilon(w+w'),
\]
contradicting $w+w'>0$ and $\epsilon<1/2$. \qed
\end{proof}

The two lemmas bound the tree's cost by
$1/(1-\epsilon)$ times the total weight of the spanner.
Together with the light-spanner bound, this completes
the proof of Theorem~\ref{thm:mst-doubling}.

\mstdoubling*

\begin{proof}
Let $T$ be the tree returned by
Algorithm~\ref{alg:mst-doubling}, and fix any
$\Metric\in\ConsMetrics{\Profile}$ of doubling
dimension $\Dim$.
By Lemma~\ref{lem:ecs-construction}, $T$ is endpoint
cut-stable with respect to $\Metric$.

Set $\epsilon=1/4$ and choose the corresponding
light spanner $\Spanner$.
By Lemmas~\ref{lem:single-edge} and~\ref{lem:injective},
each positive-weight edge of $T$ is assigned to a
distinct edge of $\Spanner$ whose weight is at least
$1-\epsilon$ times its own. Summing gives
\[
\cost_{\Metric}(T)
\le
\frac{\cost_{\Metric}(\Spanner)}{1-\epsilon}
\le
\frac{\epsilon^{-O(\Dim)}}{1-\epsilon}
\cost_{\Metric}(\MST_G)
\le
2^{O(\Dim)}\cost_{\Metric}(\MST_G).
\]
Since this holds for every consistent metric of
doubling dimension $\Dim$, the algorithm has
distortion $2^{O(\Dim)}$ on this class. \qed
\end{proof}

%% file: Files/Query.tex
\section{General Metrics With Edge-Comparison Queries}
\label{sec:mst-query}

The preceding sections use structural assumptions, such as tree metrics
and bounded doubling dimension, to obtain distortion guarantees independent
of $\NumVertices$. We now return to general metrics and ask how much additional
information is sufficient to approach constant distortion. We show that a
linear number of edge-comparison queries suffices to obtain distortion
$O(\log^{\ast} \NumVertices)$.

An edge-comparison query takes two edges $(\firstV,\secondV)$ and $(\thirdV,\fourthV)$ and reveals
whether
\[
    \Metric(\firstV,\secondV)\leq \Metric(\thirdV,\fourthV).
\]
It does not reveal either numerical distance. Comparisons between edges
sharing an endpoint can be ordered consistently with their lengths
using the ordinal profile. The
additional queries allow us to compare candidates with different endpoints.

Algorithm~\ref{alg:mst-query} applies a prescribed number of exact
Bor\r{u}vka rounds to successively smaller representative sets.
To describe this construction, we first explain how to perform one exact
Bor\r{u}vka round on a representative set $R\subseteq\Vertices$.
We use the globally best outside vertex we had, restricted
to the currently available vertices:
\[
    \GB^{R}_{\Component}(\firstV)
    =\argmin_{\secondV\in R\setminus\Component}
       \rank_{\firstV}(\secondV),
    \qquad \firstV\in\Component\subsetneq R.
\]

\begin{lemma}
\label{lem:query-exact-round}
Let $\mathcal{C}$ be the components of a forest on $R$, with
$c=|\mathcal{C}|>1$. One exact Bor\r{u}vka round can be implemented using
at most $|R|-c$ edge-comparison queries. The number of components after
the round is at most $\lfloor c/2\rfloor$.
\end{lemma}
\begin{proof}
For each component $\Component\in\mathcal{C}$, consider
\[
    E_{\Component}
    =\bigl\{(\firstV,\GB^{R}_{\Component}(\firstV)):
              \firstV\in\Component\bigr\}.
\]
For every $\firstV\in\Component$ and
$\secondV\in R\setminus\Component$, consistency implies
\[
    \Metric(\firstV,\GB^{R}_{\Component}(\firstV))
    \leq\Metric(\firstV,\secondV).
\]
Thus a minimum-length edge leaving $\Component$ can be found among
$E_{\Component}$. A tournament selects such an edge with at most
$|\Component|-1$ comparisons. As in the upper-bound algorithm, write
this proposal as $(r(\Component),g_{\Component})$, where
$g_{\Component}=\GB^{R}_{\Component}(r(\Component))$.
Here $r(\Component)$ is the endpoint of the minimum candidate selected
by the tournament.

All proposals are computed using the partition at the beginning of the
round. Adding a maximal subset that preserves acyclicity implements an
exact Bor\r{u}vka round. The total number of queries is at most
\[
    \sum_{\Component\in\mathcal{C}}(|\Component|-1)=|R|-c.
\]
The maximality argument in Lemma~\ref{lem:mst-round} applies here as well:
no old component can remain isolated, since its proposal could otherwise
be added without creating a cycle. Every new component therefore contains
at least two old components. \qed
\end{proof}

Applying these rounds until completion gives the classical Bor\r{u}vka
algorithm~\cite{boruuvka1926jistem}. By its cut-property argument, every
intermediate forest extends to an \MST of the graph on which the rounds
are performed. Starting from $R=\Vertices$, at most
$\lceil\log_2\NumVertices\rceil$ rounds are needed, each using at most
$\NumVertices$ comparisons. Consequently, this implementation computes
an exact \MST using $O(\NumVertices\log\NumVertices)$ edge-comparison queries.

To reduce the number of queries to $O(\NumVertices)$, we limit the number
of rounds before restricting the computation to a smaller representative
set. Define the round budgets by
\begin{equation}
    \Budget_0=1,\qquad \Budget_{\ell+1}=2^{\Budget_\ell}\quad(\ell\geq0).
    \label{eq:query-budgets}
\end{equation}
In outer iteration $\ell$, run at most $\Budget_\ell$ exact Bor\r{u}vka rounds
on $R_\ell$, stopping early if the forest $F_\ell$ becomes connected.
Then retain one representative from each component of $(R_\ell,F_\ell)$
to obtain $R_{\ell+1}$. Representatives are selected by a fixed vertex
order, and comparisons and cycle elimination are processed in fixed
orders, so the algorithm is deterministic.

\begin{algorithm}[H]
\caption{MST with edge-comparison queries}
\label{alg:mst-query}
\begin{algorithmic}[1]
\Require Ordinal profile $\Profile=(\pi_{\firstV})_{\firstV\in\Vertices}$ and edge-comparison query access
\Ensure A spanning tree $F$ on $\Vertices$
\State $F\gets\emptyset$, $R_0\gets\Vertices$, $\ell\gets0$, $\Budget_0\gets1$
\While{$|R_{\ell}|>1$}
    \State $F_{\ell}\gets\emptyset$
    \For{$s=1,\ldots,\Budget_{\ell}$}
        \State $\mathcal{C}\gets$ components of $(R_{\ell},F_{\ell})$
        \If{$|\mathcal{C}|=1$}
            \State \textbf{break}
        \EndIf
        \State $E\gets\emptyset$
        \For{each component $\Component\in\mathcal{C}$}
            \State $E_{\Component}\gets
                \{(\firstV,\GB^{R_\ell}_{\Component}(\firstV)):
                   \firstV\in\Component\}$
            \State Find a minimum-length edge of $E_{\Component}$ by comparisons
            \State $r(\Component)\gets$ its endpoint in $\Component$
            \State $g_{\Component}\gets
                \GB^{R_\ell}_{\Component}(r(\Component))$
            \State $E\gets E\cup\{(r(\Component),g_{\Component})\}$
        \EndFor
        \State Choose, in fixed order, a maximal $A\subseteq E$
        such that $F_\ell\cup A$ is acyclic
        \State $F_\ell\gets F_\ell\cup A$
    \EndFor
    \State $F\gets F\cup F_{\ell}$
    \State $R_{\ell+1}\gets$ one representative, chosen by a fixed vertex order,
    from each component of $(R_{\ell},F_{\ell})$
    \State $\Budget_{\ell+1}\gets2^{\Budget_{\ell}}$, $\ell\gets \ell+1$
\EndWhile
\State \Return $F$
\end{algorithmic}
\end{algorithm}

At the beginning of each outer iteration, every component of the accumulated forest
$F$ contains exactly one vertex of $R_{\ell}$. Consequently, adding $F_{\ell}$
creates no cycle in $F$, and selecting one representative per component
of $(R_{\ell},F_{\ell})$ preserves this invariant. When a single representative
remains, $F$ is connected and is therefore a spanning tree. All distances
used on $R_{\ell}$ are restrictions of the original metric; we do not replace
them by distances in the accumulated forest.

For the analysis, fix an ordinal profile $\Profile$ and a metric
$\Metric\in\ConsMetrics{\Profile}$. We start with a simple lemma bounding
the cost of an \MST on an induced subgraph by the cost of the \MST on the
original vertex set.

\begin{lemma}
\label{lem:query-subset-mst}
For every nonempty $R\subseteq\Vertices$, let $G[R]$ be the complete
metric graph induced by $R$. Then
\[
    \cost_{\Metric}(\MST_{G[R]})
    \leq 2\cost_{\Metric}(\MST_G).
\]
\end{lemma}
\begin{proof}
The claim is immediate when $|R|=1$. Otherwise, double the edges of
$\MST_G$ and traverse an Euler tour of the resulting multigraph.
Shortcut this walk to visit the vertices of $R$ in their order of first
appearance, returning to the first vertex. By the triangle inequality,
the resulting closed walk has length at most
$2\cost_{\Metric}(\MST_G)$. Deleting its closing edge gives a spanning
path on $R$ of no greater cost. The minimum spanning tree on $R$ is
no more expensive than this path. \qed
\end{proof}

This bound yields Lemma~\ref{lem:query-stage-cost}, which relates the
cost of the output to the number $h$ of outer iterations, each consisting
of a budgeted Bor\r{u}vka computation followed by a reduction to
representatives.

\begin{lemma}
\label{lem:query-stage-cost}
If Algorithm~\ref{alg:mst-query} executes $h\geq1$ outer iterations, then
\[
    \cost_{\Metric}(F)
    \leq(2h-1)\cost_{\Metric}(\MST_G).
\]
\end{lemma}
\begin{proof}
For each $\ell$, the forest $F_\ell$ is constructed by exact Bor\r{u}vka
rounds on $G[R_\ell]$ and therefore extends to an \MST of $G[R_\ell]$.
Hence
\[
    \cost_{\Metric}(F_\ell)
    \leq\cost_{\Metric}(\MST_{G[R_\ell]}).
\]
Since $R_0=\Vertices$, the first forest satisfies
$\cost_{\Metric}(F_0)\leq\cost_{\Metric}(\MST_G)$.
For every $\ell\geq1$, Lemma~\ref{lem:query-subset-mst} gives
$\cost_{\Metric}(F_\ell)\leq2\cost_{\Metric}(\MST_G)$.
Summing over the $h$ outer iterations gives
\[
    \cost_{\Metric}(F)
    \leq\sum_{\ell=0}^{h-1}\cost_{\Metric}(F_\ell)
    \leq\bigl(1+2(h-1)\bigr)\cost_{\Metric}(\MST_G).
\]
\qed
\end{proof}

We next bound how many times the budgeted computation is repeated.
Every exact Bor\r{u}vka round halves the number of components, and replacing
each component by a representative preserves this reduction. Thus a total
of $\lceil\log_2\NumVertices\rceil$ completed rounds is sufficient for
termination. The budget recurrence reaches this total within
$\log_2^*\NumVertices$ outer iterations, as shown next. Recall that
$
    \log_2^*\NumVertices
    =\min\{q\geq0:\log_2^{(q)}\NumVertices\leq1\},
$
where $\log_2^{(0)}\NumVertices=\NumVertices$.

\begin{lemma}
\label{lem:query-iteration-count}
For $\NumVertices\geq2$, Algorithm~\ref{alg:mst-query} executes at most
$\log_2^*\NumVertices$ outer iterations.
\end{lemma}
\begin{proof}
Let $m_\ell=|R_\ell|$ and
\[
    S_\ell=\sum_{a=0}^{\ell-1}\Budget_a,\qquad S_0=0.
\]
Whenever outer iteration $\ell$ is reached, every preceding iteration
has completed its full budget; otherwise the algorithm would already
have terminated. The halving property therefore implies
\begin{equation}
    m_\ell\leq\frac{\NumVertices}{2^{S_\ell}}.
    \label{eq:query-representative-count}
\end{equation}

Put $q=\log_2^*\NumVertices$. By the definition of the iterated logarithm
and Equation~\eqref{eq:query-budgets},
\[
    \NumVertices\leq \Budget_q=2^{\Budget_{q-1}}.
\]
Consequently,
\[
    S_q=\sum_{\ell=0}^{q-1}\Budget_\ell
    \geq \Budget_{q-1}
    \geq\lceil\log_2\NumVertices\rceil.
\]
If the algorithm has not stopped earlier, completing these $q$ budgets
leaves at most $\NumVertices/2^{S_q}\leq1$ representative. Thus
$h\leq q=\log_2^*\NumVertices$. \qed
\end{proof}

It remains to establish the linear query bound. Although the round
budgets increase, each budget is applied to a smaller representative
set. The following lemma shows that the total number of additional
comparisons is less than $\NumVertices$.

\begin{lemma}
\label{lem:query-total}
Algorithm~\ref{alg:mst-query} uses fewer than $\NumVertices$
edge-comparison queries.
\end{lemma}

\begin{proof}
The claim is immediate for $\NumVertices=1$.
Assume $\NumVertices\geq2$, and let $h$ be the number of outer
iterations. Retain the notation $m_\ell$ and $S_\ell$ from the
preceding proof, and let $Q_\ell$ denote the number of queries
in outer iteration $\ell$.

Since $F_\ell$ is initially empty, every component in the first
round is a singleton with exactly one candidate edge. This round
therefore requires no queries. By Lemma~\ref{lem:query-exact-round},
each remaining round uses at most $m_\ell$ queries. Hence
\[
    Q_\ell\leq m_\ell(\Budget_\ell-1),
\]
including when the iteration stops early.
In particular, $Q_0=0$ because $\Budget_0=1$.

For every executed iteration $\ell\geq1$,
Equation~\eqref{eq:query-representative-count} gives
\[
\begin{aligned}
    Q_\ell
    &\leq \frac{\NumVertices(\Budget_\ell-1)}{2^{S_\ell}}=\NumVertices\left(
        \frac{1}{2^{S_{\ell-1}}}
        -\frac{1}{2^{S_\ell}}
      \right),
\end{aligned}
\]
where the equality uses
$\Budget_\ell=2^{\Budget_{\ell-1}}$ and
$S_\ell=S_{\ell-1}+\Budget_{\ell-1}$.

If $h=1$, then $Q=Q_0=0<\NumVertices$.
Otherwise, summing and using $S_0=0$ yields
\[
\begin{aligned}
    Q
    &=\sum_{\ell=1}^{h-1}Q_\ell\\
    &\leq\NumVertices
      \sum_{\ell=1}^{h-1}
      \left(
        \frac{1}{2^{S_{\ell-1}}}
        -\frac{1}{2^{S_\ell}}
      \right)\\
    &=\NumVertices\left(1-\frac{1}{2^{S_{h-1}}}\right)
    <\NumVertices.
\end{aligned}
\]
\qed
\end{proof}

Combining Lemmas~\ref{lem:query-stage-cost},
\ref{lem:query-iteration-count}, and~\ref{lem:query-total}, we obtain,
for every $\NumVertices\geq2$,
\[
    Q<\NumVertices,
    \qquad
    \cost_{\Metric}(F)
    \leq(2\log_2^*\NumVertices-1)\cost_{\Metric}(\MST_G).
\]
And also for $\NumVertices=1$ is trivial and the empty tree is optimal and requires no queries.
This proves the previously stated theorem.

\mstquery*

%% file: Files/Conclusion.tex
\section{Conclusion}
We established lower and upper bounds for ordinal \MST{} distortion
and explored how structural restrictions and additional information
improve the guarantees. Our algorithms achieve logarithmic distortion
on general metrics, exact optimality on tree metrics realized on the
given vertex set, and constant distortion for fixed doubling dimension
through endpoint cut stability. A linear number of edge-comparison
queries further gives near-constant distortion on general metrics.

The most immediate open question is to close the gap between the
general upper and lower bounds, through a stronger lower bound or
a sharper ordinal algorithm. In the query model, an interesting
direction is to achieve constant distortion with only
$o(\NumVertices)$ comparisons. It also remains to determine whether
randomization improves the guarantees, using new algorithms or
stronger lower-bound constructions.

For doubling metrics, we leave open whether the $2^{O(\Dim)}$ bound
is tight for endpoint cut-stable trees, whether another ordinal
algorithm can improve its dependence on $\Dim$, and whether a
matching ordinal lower bound exists. More broadly, further metric
restrictions and structural properties of ordinal profiles may
permit stronger guarantees, potentially bridging the gap between
logarithmic distortion on general metrics and exact optimality
on tree metrics.